\documentclass{article}

\usepackage{amssymb}
\usepackage[a4paper, margin=0.8in]{geometry}

\usepackage{enumerate}
\usepackage{graphicx}%
\usepackage{multirow}%
\usepackage{amsmath,amssymb,amsfonts}%
\usepackage{amsthm}%
\usepackage{mathrsfs}%
\usepackage[title]{appendix}%
\usepackage{xcolor}%
\usepackage{textcomp}%
\usepackage{manyfoot}%
\usepackage{booktabs}%
\usepackage{algorithm}%
\usepackage{algorithmicx}%
\usepackage{algpseudocode}%
\usepackage{listings}%
\usepackage{bbm}
\usepackage{float}
\usepackage{placeins}
\usepackage{authblk}
\usepackage{doi}
\nocite{*}
\usepackage[numbers,sort&compress]{natbib}
\title{Random mixtures in Bayes Hilbert spaces}

\author[1]{Giulia Patan\`e\thanks{Corresponding author: \texttt{giulia.patane@hu-berlin.de}}}
\author[2]{Sonja Greven}
\author[1]{Alessandra Menafoglio}

\affil[1]{MOX, Department of Mathematics, Politecnico di Milano, Via Edoardo Bonardi 9, 20133 Milan, Italy}
\affil[2]{Chair of Statistics, School of Business and Economics, Humboldt-Universit\"at zu Berlin, Spandauer Stra\ss e 1, 10117 Berlin, Germany}

\date{}

\begin{document}

\newtheorem{theorem}{Theorem}
\newtheorem{corollary}[theorem]{Corollary}
\newtheorem{proposition}[theorem]{Proposition}
\newtheorem{lemma}[theorem]{Lemma}
\theoremstyle{definition}
\newtheorem{definition}{Definition}[section]
\theoremstyle{remark}
\newtheorem{remark}{Remark}
\maketitle

\begin{abstract}
We present a framework for the analysis and unmixing of random density mixtures in the Bayes Hilbert space. General identifiability results for mixtures in Hilbert spaces are established and applied to the Bayes Hilbert space setting. Building on these results, we propose a penalised maximum likelihood approach for the unmixing of Bayes Hilbert mixtures aimed at recovering the statistically space-efficient representation, together with a computationally efficient coordinate-wise maximisation algorithm for its implementation. The methodology is illustrated through a hyperspectral data application, where observations can be naturally embedded in the Bayes Hilbert space and analyzed in terms of distributional shape rather than amplitude. A complementary simulation study demonstrates the interpretability and practical performance of the proposed approach.
\end{abstract}


\vspace{1em}








\section{Introduction}
\label{intro}

Probability densities are central to statistics because they describe how values of a variable are distributed across a population. In many modern applications, the unit of analysis is not a single individual but an entire society or complex system. In these cases, the overall distribution can reveal important patterns that simple summaries, like the mean or variance, cannot. For example, two completely different distributions can have the same mean and variance, yet behave in fundamentally different ways. Capturing such differences requires looking at the full density, rather than just a few summary statistics.

In this work, we introduce and explore mixtures of densities, decomposing an observed density into interpretable constituent parts with compositional proportions. We show that, under mild assumptions, the representation whose vertices lie on the support of the data is identifiable, and that a representation
maximising the statistical spread of the compositional weights always exists. Moreover, we propose a penalised maximum likelihood unmixing strategy and validate the model through simulations and on a case study -- using data from AVIRIS Indian Pines \citep{PURR1947}.

Working with densities is not straightforward because they are constrained objects: they cannot be negative, and their total mass must always sum to one. Classical approaches in functional data analysis (FDA, \cite{ramsaysilvermanbook}) generally assume that the data can vary freely, which is not true for densities. Applying these methods directly to densities can lead to misleading or invalid results.

To overcome this problem, researchers developed the concept of the \textit{Bayes Hilbert space} \citep{VanDenBoogaart2014}, a mathematical 
framework specifically designed for densities. In this space, densities can be combined and manipulated in ways that respect their natural 
constraints. Indeed, transformations exist that allow traditional functional data analysis techniques to be applied safely, while making 
the study of densities simpler.

Embedding densities in the Bayes Hilbert space has already led to many practical advances. Researchers have used it to handle higher dimensional domains \citep{hron2023, hron25}, to work with non-Lebesgue reference measures \citep{talska20}, to model spatially dependent densities \citep{menafoglio2016}, to analyze situations where only sparse observations are available \citep{steyer2023, eva_sparsedensity2025}, and where densities are the predictors \citep{Jaskova2023} or the responses \citep{eva2025, eva_sparsedensity2025, keilbar2026} of regression models.

Yet, one fundamental problem remains unaddressed within this framework: the unmixing of densities. Mixtures arise naturally in many fields, representing populations composed of multiple subgroups or overlapping latent processes. Consider, for instance, the income distribution of a country composed of distinct socioeconomic groups: the observed density is a mixture of the group-specific densities, and unmixing it reveals the underlying subpopulation structure. Addressing this problem within the Bayes Hilbert space framework, however, is not straightforward. While simplicial PCA \citep{Hron16, steyer2023} offers the closest existing tool for decomposing densities in the Bayes Hilbert (BH) space, it does so along orthogonal directions of variation -- not as a mixture of components with compositional weights. Moreover, such a decomposition cannot be equated with classical probabilistic mixtures, which are linear combinations in the traditional sense and do not respect the geometric structure of the Bayes Hilbert space. This distinction matters whenever the population is heterogeneous and the goal is to recover latent subgroups or overlapping processes. As we discuss in Section~\ref{BHunmixing}, BH mixtures induce a genuinely different notion of mixing, and their interpretation requires care. Indeed, classical mixtures induce non-linear structures in Bayes Hilbert spaces. Conversely, mixtures defined within Bayes Hilbert spaces are given by convex combinations of distributions with respect to the underlying vector space operations, and therefore do not, in general, coincide with classical mixtures. This discrepancy has been partially addressed in the literature on $\alpha$-mixtures, which introduce alternative, generally non-linear, combination rules \citep{Shojaee2021, ShojaeeAsadiFinkelstein2022, AsadiEbrahimiSoofi2019, ShojaeeMomeni2023}. These mixtures have been investigated in several applied areas, including survival analysis and reliability modeling, stochastic ordering and comparison of heterogeneous populations, and Bayesian inference for lifetime models, where they provide flexible interpolation schemes between arithmetic -- for $\alpha=1$ -- and geometric -- $\alpha\to 0^+$ -- aggregation of probability models. However, this line of work does not provide a practical tool for nonparametric estimation of the vertices and of the proportions' distribution. In this work, we focus instead on exploring Bayes Hilbert mixtures as a powerful alternative grounded in the literature on compositional data analysis \citep{VanDenBoogaart2014, hron2023, hron25, talska2018, talska20, steyer2023, eva_sparsedensity2025, Jaskova2023, eva2025, keilbar2026}.

A particularly natural application of this framework arises in hyperspectral imaging, where each pixel records a superposition of pure-material signatures. Hyperspectral data consist of reflectance -- or radiance -- curves measured across a dense grid of wavelengths, forming a high-dimensional functional signature for each pixel of the recorded scene. In this context, it is not the absolute magnitude of reflectance that is most informative, but rather the relative distribution of energy across wavelengths -- that is, the overall shape of the spectral curve. Subtle variations in peaks, slopes, and absorption features reveal the presence of specific materials, surface properties, or ongoing chemical processes \citep{Mazdeyasna2025}. For this reason, hyperspectral reflectance curves are usually analysed after normalisation, encoding the relative distribution of energy across wavelengths rather than absolute intensity, and can thus be treated as densities within the Bayes Hilbert space framework -- making unmixing a problem of recovering compositional proportions over a simplex of endmember densities. More broadly, because densities are positive, compactly supported, and continuous, they can represent any functional object whose key feature is its shape rather than its amplitude, making this framework widely applicable beyond classical probabilistic settings.

In hyperspectral data, pure material exhibits a characteristic signature -- either in reflectance or radiance -- allowing its remote identification within a scene, namely a spatially extended domain such as a planetary surface portion \citep{Vasile2024, alfre2025}. When a scene is captured by remote sensors, every pixel of the resulting image corresponds to a hyperspectral observation. However, hyperspectral measurements rarely coincide exactly with pure-material signatures, as sensing technologies introduce spatial blurring and mixing effects due to neighbouring light sources and limited spatial resolution. For this reason, a dedicated research area known as \textit{Hyperspectral Unmixing} (HU) has emerged, aiming to develop methods that estimate the proportions -- or, in some formulations, the probabilities of occurrence -- of each material contributing to the spectrum recorded at a given pixel. Moreover, since the set of materials present in a scene is not always known \emph{a priori}, the pure-material signatures -- endmembers -- are often estimated in a fully unsupervised manner. Some approaches, such as Vertex Component Analysis (VCA, \cite{Nascimento2005}), adopt the so-called \textit{pure-pixel assumption}, which postulates that at least one pixel in the scene corresponds to each pure material; however, this assumption may fail in practice, potentially leading to inaccurate or non-identifiable endmember estimates. Non-negative Matrix Factorisation (NMF, \cite{hyper_pen,Arngren2013,rajabi2013_hyperspectral}), which identifies the minimum-volume simplex enclosing the data, and Minimum-volume HU \citep{Jun2008} relax this assumption, but remain sensitive to outliers: extreme observations can artificially inflate the enclosing simplex and distort the estimated endmembers. Our approach occupies a middle ground: rather than requiring pure pixels to be present in the dataset, we assume that the endmembers have positive probability density -- that is, they need not be observed, but must be plausible under the data-generating distribution. Under this assumption, the tightest convex hull containing the data is identifiable; to estimate it, we penalise representations in which the
compositional weights are poorly spread over the simplex, a criterion based on the covariance structure of the weights rather than on extreme
observations alone.

The article is organised as follows. In Section~\ref{BHunmixing} we formally define and study random density mixtures within the Bayes 
Hilbert space, and discuss the difference with respect to probabilistic mixtures. Section~\ref{identifiability} investigates the identifiability 
conditions of the corresponding unmixing problem in generic Hilbert spaces; Section~\ref{pmlu} presents the proposed maximum likelihood 
estimation procedure in Bayes Hilbert spaces. The practical performance and applicability of the proposed methodology are demonstrated through 
an extensive simulation study (Section~\ref{sec:simulations}) -- with synthetic probability density data -- and a real-data application to 
hyperspectral imagery (Section~\ref{sec:aviris}).


\section{Theoretical background}
\label{sec:theory}
\subsection{Bayes Hilbert Space}
\label{bayes_hilbert_space}
A Bayes Hilbert space on a compact interval \( I \subset \mathbb{R} \), denoted by \( B^2(I) \), is an infinite-dimensional Hilbert space specifically designed for the analysis of density functions \citep{VanDenBoogaart2014}. Formally, \( B^2(I) \) can be defined as the quotient space 
\[
L_+^2(I) / =_{B^2},
\]
where 
\[
L_+^2(I) = \{ f: I \rightarrow \mathbb{R^+} \textrm{ s.t. } \log(f)\in L^2  \},
\]
and \( =_{B^2}\) is the equivalence relation given by
\[
f =_{B^2} g \quad \text{if and only if} \quad f = \eta \, g \ \text{ a.e. on } I \text{, for some constant \( \eta >0 \).}
\]

Thus, elements of \( B^2(I) \) are equivalence classes of positive  functions whose logarithm is square integrable, that differ only by a positive multiplicative constant. In other words, functions are identified up to scale, reflecting the fact that only their relative structure -- or shape -- is relevant. By convention, in the following, each equivalence class is represented, if possible, by its normalised element, i.e. the density integrating to one over \( I \).
As a Hilbert space, it is equipped with a sum \(\oplus\) (a.k.a. perturbation), a multiplication by scalar \(\odot\) (a.k.a. powering), and a scalar product, with respect to which \(B^2(I)\) is complete. For \(f,g \in B^2(I)\) and \(\alpha \in \mathbb{R}\), the operations are defined as
\begin{equation}
\begin{cases}
f \oplus g = \dfrac{f \cdot g}{\int_I f(s) \cdot g(s) \, ds},\\[0.5em]
\alpha \odot f = \dfrac{f^\alpha}{\int_I f(s)^\alpha \, ds}.
\end{cases}
\label{operations}
\end{equation}

To define a scalar product on \(B^2(I)\), \cite{VanDenBoogaart2014} introduced the \textit{centered log-ratio (clr) transformation}, which maps densities into elements of
\[
L^2_0(I) = \Big\{ f \in L^2(I) \;\big|\; \int_I f(s) \, ds = 0 \Big\}.
\]

The clr-transformation is defined as
\begin{equation}
\operatorname{clr}(f)(t) = \log(f(t)) - \frac{1}{\lambda(I)} \int_I \log(f(s)) \, ds,
\end{equation}
where \(\lambda\) denotes the Lebesgue measure on \(I\). By construction, with $\langle f, g\rangle_{B^2(I)}=\langle \operatorname{clr}(f),\operatorname{clr}(g)\rangle_{L^2(I)}$, the clr-transformation is an isometry between \(B^2(I)\) and \(L^2_0(I)\), and it preserves the Hilbert space structure. Moreover, \(B^2(I)\) is separable \citep{VanDenBoogaart2014}, a property that is particularly important in statistical and computational contexts, where data must be efficiently represented and recorded.

This formulation is here considered for ease of exposition, although a wide range of extensions exists, including generalizations to more complex or unbounded supports and non-Lebesgue reference measures \citep{VanDenBoogaart2014, hron2023, talska20}.

\subsection{Linear mixtures}
Let $\{ h_j\}_{j=1}^m\subset B^2(I)$ be a set of densities.
A linear mixture is defined as \citep{FruhwirthSchnatter2006}
\[
g=\sum_{j=1}^m p_j\, h_j,
\qquad \sum_{j=1}^m p_{j}=1,\ p_j\geq 0\ \forall j.  
\]
If we denote the vector of proportions $\boldsymbol{p}=(p_j)_j$, then by construction $\boldsymbol{p}\in S^m$, where $S^m\subset\mathbb{R}^m$ is the $(m-1)$ dimensional simplex. Thus, \(g\) is a linear convex combination of the densities \(\{h_j\}_{j=1}^m\), or equivalently a weighted arithmetic mean, where the weights sum to 1.
Let \(X\sim g\). Then, \(X\) has the same distribution as the mixture \(Y\), defined conditionally as \(Y|\{J=j\}\sim h_j\) for all $j\in \{1,\ldots,m\}$, with mixing index \(J\sim\boldsymbol p\) supported on \(\{1,\ldots,m\}\).

Consequently, the linear mixture provides an appropriate model when each
observation is assumed to be generated by first selecting a component
density \(h_j\) with probability \(p_j\), and then sampling from that
density.
In this sense, the linear mixture represents a purely
\emph{probabilistic} mixture, as it explicitly models heterogeneity
through a random selection among distinct components.

In the next section, we introduce a framework for random Bayes Hilbert mixtures, highlighting the crucial differences with respect to the classical linear -- probabilistic -- mixture.

\section{Bayes Hilbert mixtures}
\label{BHunmixing}
We now introduce the core theoretical developments of our work, beginning with the definition of Bayes Hilbert mixtures. Note that this construction does not depend on the specific geometry of Bayes spaces, but is valid in any Hilbert space. However, for the purposes of this work, we restrict our attention to mixtures defined within a Bayes Hilbert space and adopt the corresponding notation throughout. 
Let $\{ h_j\}_{j=1}^m\subset B^2(I)$ be a set of linearly independent densities, where linearity has to be intended in the Bayes Hilbert sense, namely with respect to the operators $\oplus,\odot$ defined above. A \textit{convex hull} in the Bayes Hilbert space $(B^2(I),\oplus,\odot)$  is a set
\[
\mathcal M( \{h_j\}_{j=1}^m)=\bigg\{g\in B^2(I) \bigg| g=\bigoplus_{j=1}^m  p_j \odot  h_j,
\sum_j p_j=1, \ p_j\geq 0 \ \forall j\bigg\}.\]

In the following, we name \textit{Bayes Hilbert mixtures} the elements of $\mathcal{M}$. Note that, by projecting $\mathcal M$ through $\operatorname{clr}$ into the $L_0^2(I)$ space, the image is $conv(\operatorname{clr}(h_1),...,\operatorname{clr}(h_m))$, 
i.e. the classical (linear) convex hull of $\operatorname{clr}$-transformed densities $\{ h_j\}_{j=1}^m$. 

Let
\[
g=\bigoplus_{j=1}^m p_j \odot  h_j =\frac{\prod_{j=1}^m{h_j^{p_j}}}{\int_I \prod_{j=1}^m{h_j^{p_j}(x)dx}}=_{B^2} \prod_{j=1}^m{h_j^{p_j}} ,
\]
for a given set of $\{h_j\}_{j=1}^m$ in $B^2(I)$ and $\boldsymbol{p}=(p_j)_j$ in $S^m$. By construction, the resulting density \(g\) corresponds to the weighted geometric mean of the densities \(\{h_j\}_{j=1}^m\). Consequently, the Bayes Hilbert mixture represents a form of aggregation that differs fundamentally from probabilistic mixing: rather than describing a population composed of
distinct subpopulations, it yields a single distribution that integrates the characteristics of the original densities.
In this sense, the Bayes Hilbert mixture is more appropriately viewed as a \emph{fuzzy} mixture than as a probabilistic one. 

The concept of a fuzzy mixture is  widely recognized \citep{FGMM}: in general, it describes a framework in which each observation does not belong exclusively to a single component, but rather to all of them simultaneously, with varying degrees of membership. Just as \emph{grey} is not simply a mixture of white and black pixels, but a shade in its own right -- where white and black represent the two extremes of a continuum -- the Bayes 
Hilbert mixture yields a genuinely new density that interpolates smoothly between the original components. It is strictly related to the problem of soft-clustering, as 
it models data belonging  to more than one cluster at the same time, with different proportions. This seems appropriate e.g.\ in our application on hyperspectral unmixing, where pixels usually contain more than one material at the same time.

Figure \ref{arit_geom} (top) displays an example with Beta distributions. The two continuous curves represent two densities $d_1, d_2$ in the Bayes Hilbert space $B^2([0,1])$. The dashed curve represents the result of a linear mixture, while the dotted one represents the result of a Bayes Hilbert mixture. In both cases, $\boldsymbol{p}=(0.5,0.5)$. The linear mixture is bimodal, with modes that are near the modes of the two original densities. 
By contrast, the Bayes Hilbert mixture is unimodal, with the mode in $0.5$: this highlights that the two kinds of mixtures can be strongly different. For example, let us consider the densities and mixtures of Figure \ref{arit_geom} (top) having a domain equal to a color map from white, 0, to black, 1. Figure \ref{arit_geom} (bottom) displays pixels that are i.i.d sampled from $d_1$ (first from left, mostly light), $d_2$ (last, mostly dark), their linear mixture (second, heterogeneous mixture of light and dark pixels) and Bayes Hilbert mixture (third, mostly \textit{grey}). This illustrates that linear mixtures are  preferable when the data are believed to arise from separable subpopulations (here light and dark pixels), whereas Bayes Hilbert mixtures define a population, where mixing takes place within each observation (here within pixels, leading to differently grey pixels). The choice between the two mixing paradigms is thus driven by the setting and inferential goal.

\begin{figure}[H]
    \centering
    \includegraphics[width=0.4\linewidth]{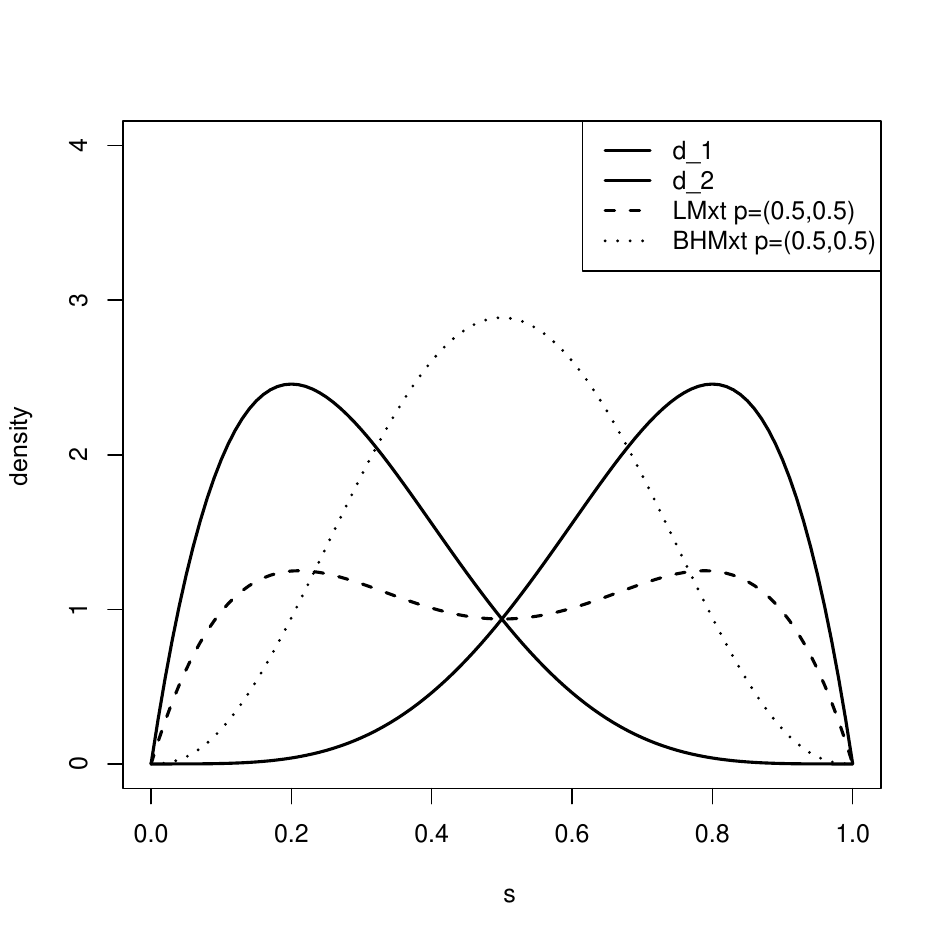}\\[1em]
    \centering\includegraphics[width=0.8\linewidth]{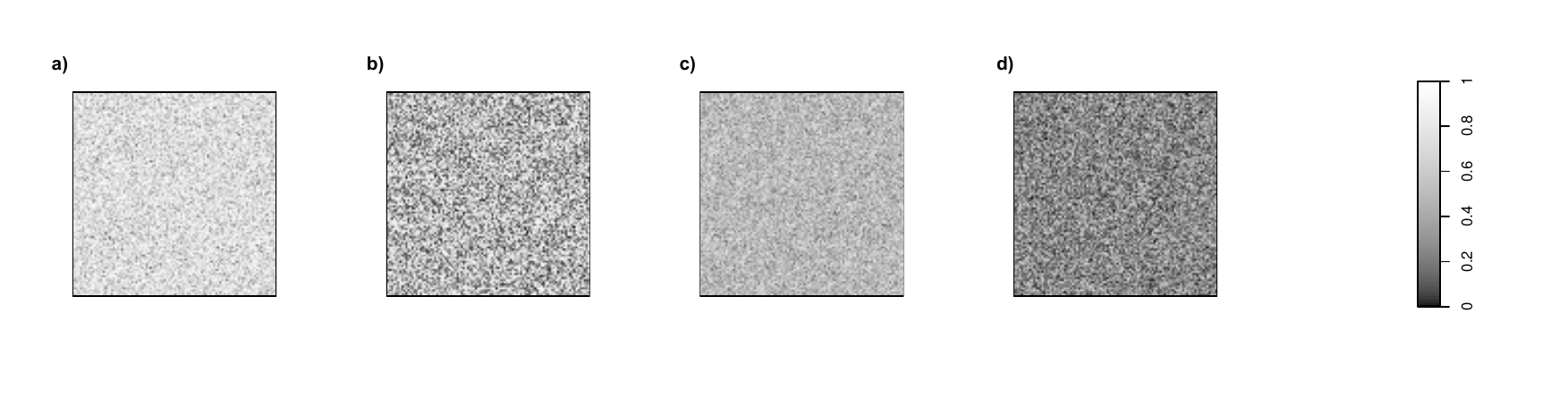}
    \caption{(Top) Density curves of Beta($2,5$) and Beta($5,2$) distributions (\texttt{d\_1} and \texttt{d\_2}), their linear (uniform) mixture (dashed), and their Bayes Hilbert (uniform) mixture (dotted). (Bottom) Independent samples from a) $d_1$, b) linear uniform mixture, c) BH uniform mixture, and d) $d_2$, stored in a matrix for ease of visualization.}
    \label{arit_geom}
\end{figure}


Formally, let $(\Omega,\mathcal{F},\mathbb{P})$ be a probability space and 
let $G:\Omega \to B^2(I)$ be a random element in the Bayes Hilbert space 
$B^2(I)$.

\begin{definition}[Bayes Hilbert random mixture]
\label{def_mix}
Let $\{h_j\}_{j=1}^m \subset B^2(I)$ be linearly independent in the 
Bayes Hilbert space. We say that $G$ is a Bayes Hilbert random mixture 
(of $\{h_j\}_{j=1}^m$) if
\[
G(\omega) \in \mathcal{M}(\{h_j\}_{j=1}^m)
\quad \text{for all } \omega \in \Omega,
\]
where $\mathcal{M}(\{h_j\}_{j=1}^m)$ denotes the associated convex hull.
\end{definition}

So far, our formulation shows that if $G$ is any random mixture, then it \emph{can} be represented by a collection of vertices $\{h_j\}_{j=1}^m$ with an associated random proportion vector $\boldsymbol{p}=(p_1,...,p_m)'$. To make this useful in practice, we may ask whether we also have the identifiability of any such $G$, i.e. whether it can be uniquely represented as a mixture. We discuss two variants of identifiability:

\begin{enumerate}[(i)]
    \item \textbf{Partial identifiability}: given $\{h_j\}_{j=1}^m$, the distribution of $\boldsymbol{p}$ is unique;
    \item \textbf{Total identifiability}: the representation of $G$ is unique, namely, there exists a unique pair consisting of vertices and random proportions generating $G$.
\end{enumerate}

The first property holds, as shown in Proposition~\ref{prop:unique_p} (Appendix~\ref{app:identifiability}). The second property, however, does not hold in general. Given $G$, there may exist infinitely many pairs of vertices and corresponding random proportions that generate $G$ as a random mixture.

In particular, larger convex hulls containing the support of (the distribution of) $G$, indicated as $supp(G)$, correspond to low-variance distributions of $\boldsymbol{p}$, whereas smaller convex hulls containing  $supp(G)$ require more dispersed distributions of $\boldsymbol{p}$ -- see Figure \ref{intro_mix} for a schematic explanation. As the convex hull shrinks, its vertices approach $supp(G)$; equivalently, the vertices become increasingly likely to appear in a realization.

In real datasets, the vertices are typically latent and do not appear as pure realizations -- see the discussion of the pure-pixel assumption in \cite{alfre2025}. Nevertheless, one may assume that the vertices, even if not observed, could in principle occur. Formally, this amounts to seeking the representation whose vertices are closest to $supp(G)$, or equivalently, whose associated distribution of $\boldsymbol{p}$ is as spread out as possible. This is analogue to searching for the minimum-volume convex hull, as is standard in hyperspectral unmixing problems, but in a distributional perspective. We formalise this in Section~\ref{identifiability}, with proofs in Appendix~\ref{app:identifiability}.

\section{Identifiability of mixtures in Hilbert Spaces}
\label{identifiability}

We study the identifiability of the representation of a random mixture $G$ through a set of vertices $\{h_j\}_{j=1}^m$ and the distribution of the random proportion. Throughout, $G$ is assumed to be a non-degenerate random mixture, compactly supported on the convex hull of $m$ fixed vertices, with vertices restricted to a fixed finite-dimensional subspace of $B^2(I)$. Non-degeneracy means that the realisations of $G$ are not confined to a lower-dimensional affine subspace, i.e.\ $\operatorname{aff}(\operatorname{supp} G)$ has dimension $m-1$.

\paragraph{Partial identifiability}
Given linearly independent vertices $\{h_j\}_{j=1}^m$, Proposition~\ref{prop:unique_p} (proved in Appendix~\ref{app:identifiability}) guarantees that the distribution of $\boldsymbol{p}$ is uniquely determined by $G$. This allows us to define the map
\[
\phi_G \colon \mathcal{H}_G^m \to \mathfrak{D}(S^m),\]
\[
\phi_G(\{h_j\}_{j=1}^m) = D
\iff
G \stackrel{d}{=} \bigoplus_{j=1}^m p_j \odot h_j,\quad \boldsymbol{p} \sim D,
\]
where $\mathcal{H}_G^m$ is the set of $m$-tuples of linearly independent vertices whose convex hull contains $G$ almost surely (Definition~\ref{def:H^m_g}). Equipping $\mathfrak{D}(S^m)$ with the Wasserstein-1 distance (see Appendix~\ref{app:identifiability}), $\phi_G$ is continuous.

\paragraph{Minimal representation}
Let $\boldsymbol x=\psi(\boldsymbol p)\in\mathbb R^{m-1}$ denote the coordinates of $\boldsymbol p$ on its $(m-1)$-dimensional representation based on a given function $\psi:S^m\to\mathbb{R}^{(m-1)}$, and let $\Sigma_p=\operatorname{Cov}(\boldsymbol x)\in \mathbb{R}^{(m-1)\times(m-1)}$, seen as a function of the distribution of $\boldsymbol p$.  We define a strict partial order $<_G$ (Proposition~\ref{prop:strict_order}, \ref{app:identifiability}) on the set of valid representations of $G$ as follows. Since two $m$-tuples of vertices that differ only by a permutation yield the same convex hull, we work on the quotient space with respect to the equivalence relation $\sim_\pi$: $(h_1,\dots,h_m)\sim_\pi(h_1',\dots,h_m')$ if there exists a permutation $\pi$ of $\{1,\dots,m\}$ such that $(h_1,\dots,h_m)=(h'_{\pi(1)},\dots,h'_{\pi(m)})$. 
The set of valid representations is then
\[
\mathcal{C}_G
= \Bigl\{(\{h_j\}_j,\, \boldsymbol{p}) :
G \stackrel{d}{=} \textstyle\bigoplus_j p_j \odot h_j
\Bigr\}\big/{\sim_\pi}.
\]
We set $(\{h_j\}_j, \boldsymbol{p}) <_G (\{\bar h_j\}_j, \bar{\boldsymbol{p}})$ whenever $\operatorname{tr}(\Sigma_p) > \operatorname{tr}(\Sigma_{\bar p})$, where $\operatorname{tr}(\Sigma_p)$ is the total variance of the proportions on the simplex. A minimal element for $<_G$ is therefore the representation whose proportions spread out most within the convex hull, i.e.\ the statistically space-efficient one. In the estimation procedure of Section~\ref{pmlu}, this is encouraged through the penalty $\lambda_{tr}\operatorname{tr}(\Sigma_p^{-1})$ in~\eqref{hatl}, which discourages representations with small dispersion of the proportions. By contrast, directly minimising the volume of the convex hull tends to concentrate the vertices near the most extreme observations, which may distort the recovered statistical structure by overweighting atypical data points. 

The existence of a minimal element requires an assumption of tightness.
\begin{definition}[Tightness]
\label{def:tight}
$G$ is \emph{tight} if there exists $(\{h^0_j\}_j,\boldsymbol p_0)\in\mathcal C_G$, with
$\mathcal M_0:=\mathcal M(\{h^0_j\}_j)$, such that
\begin{enumerate}[(T1)]
  \item \emph{Containment:} every $(\{h_j\}_j,\boldsymbol p)\in\mathcal C_G$ with
        $\operatorname{tr}(\Sigma_p)\ge\operatorname{tr}(\Sigma_{p_0})$ satisfies
        $\mathcal M(\{h_j\}_j)\subseteq\mathcal M_0$;
  \item \emph{Monotonicity:} for every $(\{h_j\}_j,\boldsymbol p),(\{h'_j\}_j,\boldsymbol p')\in\mathcal C_G$ with
        $\mathcal M(\{h'_j\}_j)\subsetneq\mathcal M(\{h_j\}_j)\subseteq\mathcal M_0$,
        we have $\operatorname{tr}(\Sigma_{p'})>\operatorname{tr}(\Sigma_p)$.
\end{enumerate}
\end{definition}

Informally, $G$ is tight if there exists a valid convex hull $\mathcal M_0$ such that (T1) every representation at least as dispersed as $\mathcal M_0$ has its convex hull inside $\mathcal M_0$, and (T2) inside $\mathcal M_0$, shrinking a convex hull increases the trace. 

Condition (T1) alone guarantees that a minimal element exists
(Proposition~\ref{prop:min_exist}, \ref{app:identifiability}). Uniqueness
requires (T2) together with an additional assumption: we say that $G$ satisfies the \emph{Probabilistic Pure Pixel} (3P) assumption if there exists a valid representation whose vertices can all be approached with positive probability. Formally:

\begin{definition}[Probabilistic Pure Pixel assumption]
\label{def:3P}
$G$ satisfies the 3P assumption if there exists $\{h^*_j\}_{j=1}^m\in\mathcal H^m_G$ such that,
for every $j=1,\dots,m$ and every $\epsilon>0$, $\mathbb P\big(G\in D_\epsilon(h^*_j)\big)>0$,
where $D_\epsilon(h^*_j)$ is the ball centred in $h^*_j$ with radius $\epsilon$; equivalently,
$h^*_j\in\operatorname{supp}(G)$ for every $j$.
\end{definition}

This differs from the classical pure pixel assumption in Hyperspectral Unmixing, which requires
each vertex to appear as a pure observation in the sample; the 3P assumption is a
population-level requirement, which only asks every neighbourhood of each vertex to have
positive probability.

\begin{theorem}[Minimal representation]
\label{minimal_theorem}
If $G$ is non-degenerate, tight, and satisfies the 3P assumption, then $(\mathcal C_G,<_G)$ has
a unique minimal element, given by the vertices $\{h^*_j\}_{j=1}^m$ of Definition~\ref{def:3P}
paired with $\boldsymbol p^*\sim\phi_G(\{h^*_j\}_{j=1}^m)$.
\end{theorem}

The minimal element is thus the geometrically minimal representation, whose vertices belong to the support of $G$. This is crucial for interpretability: by forcing $\boldsymbol{p}$ to be spread out during estimation, we target the representation whose vertices are plausible, although possibly unobserved, realisations of $G$.

\begin{remark}[Intuition behind tightness]
\label{rem:tightness}
Let $\psi$ be the ilr transformation (Section~\ref{subsec:proportions}). If the distribution of the proportions is sufficiently represented near the vertices, the directions towards the vertices dominate the total variance $\operatorname{tr}(\Sigma_p)$. Consequently, locally around the geometrically minimal convex hull $\mathcal M^*$ of Definition~\ref{def:3P}, the total variance decreases when moving towards larger convex hulls. If, moreover, the distribution does not place heavy mass along other directions (as would happen, for instance, with strongly correlated proportions), the following typically occurs. Starting from a convex hull $\mathcal M_0$ that is already a good guess of the minimal one, it suffices to shrink it concentrically towards $\mathcal M^*$ to also attain the maximal ilr total variance, in line with (T2) in Definition~\ref{def:tight} (Figure~\ref{fig:tightness}).
\end{remark}

\section{Bayes Hilbert unmixing via penalised maximum likelihood}
\label{pmlu}
In real datasets consisting of multiple realizations of a random mixture $G$, 
the assumption that $G$ is a noiseless mixture is generally unrealistic. Without an additive noise term absorbing existing residual variability in $B^2(I)$, every observed realization would require its own vertex to be represented exactly, leading to an unbounded growth of $m$ and to overfitting of the vertices. We therefore incorporate a noise component and fix the number of components $m$ (the choice of this hyperparameter is discussed in Section~\ref{sec:tuning}). We model $G$ as
\begin{equation}
    G = G_{\mathcal M}\oplus \epsilon_{B^2} 
    \label{eq:G}
\end{equation}
where $G_{\mathcal M}$ is a noiseless mixture for which the 3P assumption holds (see Definition~\ref{def:3P}), and $\epsilon_{B^2} \in B^2(I)$ is a density-valued noise term with mean equal to the uniform distribution $[0]_{B^2}$, the neutral element of the BH space, and a covariance operator $\Gamma_\epsilon$. The perturbation term $\epsilon_{B^2}$ allows $G$ to deviate, with probability one, from the convex hull on which $G_{\mathcal M}$ lies. More precisely,
\begin{equation}
    G_{\mathcal M} := \bigoplus_{j=1}^m p_j \odot h_j,
    \label{eq:G_M}
\end{equation}
where $\{h_j\}_{j=1}^m \subset \mathcal H^m_G$ are the vertices and $\boldsymbol p := (p_j)_{j=1}^m \sim D\in \mathfrak D(S^m)$ are random proportions. The pair
\[
\big(\{h_j\}_{j=1}^m, \boldsymbol p\big) \in \mathcal H_G^m \times \mathfrak D(S^m)
\]
constitutes the minimal representation of $G_{\mathcal M}$. This representation is unique up to permutations of the indices, according to the identifiability results established in Section~\ref{identifiability}. In this section, we propose a strategy for estimating both the vertices and the distribution of the random proportions that generate $G_{\mathcal M}$, based solely on a sample of $n$ independent and identically distributed (i.i.d.) realizations of the noisy mixture $G$. 

\subsection{Modelling the proportions' distribution}
\label{subsec:proportions}
In Equation \eqref{eq:G_M}, $(p_j)_{j=1}^m=\boldsymbol p \in S^m$ denotes the vector of proportions. We assume that the isometric log-ratio (ilr,  \cite{Egozcue2003}) transformation of $\boldsymbol p$ follows a Gaussian distribution, namely
\[
\boldsymbol p \sim \psi^{-1}\!\left(\mathcal{N}_{m-1}(\boldsymbol \mu_p,\Sigma_p)\right),
\]
for some $\boldsymbol\mu_p\in\mathbb{R}^{m-1}, \Sigma_p\in \mathbb{R}^{(m-1)\times(m-1)}$ positive definite, where $\psi : S^m \to \mathbb{R}^{m-1}$ denotes the discrete ilr transformation. The ilr transformation maps the compositional vector $\boldsymbol p$ from the simplex $S^m$ to the real Euclidean space $\mathbb{R}^{m-1}$ by projecting log-ratio coordinates onto an orthonormal basis, thereby preserving distances and endowing the simplex with a standard (finite-dimensional) Hilbert space structure. Equivalently, in ilr coordinates, the model reads
\[
\boldsymbol z = \psi(\boldsymbol p) \sim \mathcal{N}_{m-1}(\boldsymbol \mu_p,\Sigma_p),
\]
and the composition is recovered through the inverse mapping
\[
\boldsymbol p = \psi^{-1}(\boldsymbol z).
\]
 This formulation makes explicit that the Gaussian assumption is imposed in Euclidean space, while the induced distribution on the simplex is ilr--normal. While other log-ratio transformations, such as the clr, are possible, they differ from the ilr by a linear transformation and thus do not affect normality.

\subsection{Modelling the density objects}
The model detailed in Equation \eqref{eq:G} and \eqref{eq:G_M} includes different density objects: the noisy realizations $\{g_i\}_{i=1}^n$, which are observable, the latent vertices $\{h_j\}_{j=1}^m$ and the noise $\epsilon_{B^2}$. 

In our work, we represent each density object through their clr-transformation, which is discretised -- applying truncation -- via finite Karhunen-Loève decomposition (i.e., simplicial functional principal component analysis, SFPCA, \cite{Hron16}) of $G$. This yields, for a generic density $g$:
\begin{equation}\label{eq:g_sfpc}g(x)=\operatorname{clr}^{-1}\bigg(\sum_{l=1}^k c_l \phi_l(x)\bigg) 
\end{equation}
where $\{\phi_l\}_{l=1}^k\in L_0^2(I)$ are the first $k$ simplicial principal components (SFPCs) -- which form a $k$-dimensional orthonormal basis -- and $c_l\in \mathbb{R}$ are the scores of $g(x)$ over the SFPCs -- i.e., $c_l = \langle g, \phi_l \rangle_{B^2} $, with $\langle \cdot, \cdot\rangle_{B^2}$ the inner product in $B^2(I)$. 

The clr-transformation of the density $\epsilon_{B^2}$ is assumed to be distributed as a Gaussian process (GP, \cite{bishop_book}), namely $\epsilon_{B^2}=\operatorname*{clr}^{-1}(\epsilon)$ with
\[
\epsilon\sim GP_{L_0^2}(0,\Gamma_\epsilon).\]
Analogous to \eqref{eq:g_sfpc}, we can expand $\epsilon=\operatorname*{clr}(\epsilon_{B^2})$  in the basis of the first $k$ eigenfunctions (SFPCs) of $G$, with coefficients then following a multivariate Gaussian distribution of dimension $k$ with mean $\mu_\epsilon=\boldsymbol0$ and covariance matrix $\Sigma_\epsilon$. 
The equivalence is shown in \cite{steyer2023}. 
Hereafter, the symbol $\boldsymbol{c}_i\in \mathbb{R}^{k}$ will denote the vector of coefficients of $\operatorname{clr}(g_i)$, associated with the $k$ elements of the basis $\{\phi_1, ..., \phi_k\}$. Analogously, the matrix $\mathbb H\in \mathbb{R}^{k\times m}$ will store the coefficients of $h_j$, $j\in\{1,...,m\}$ column-wise. For clarity, the complete model can be summarised as follows. Each observed density $g_i$ is represented by its clr-coefficient vector $\boldsymbol{c}_i \in \mathbb{R}^k$, which satisfies \begin{equation} \boldsymbol{c}_i = \mathbb{H}\boldsymbol{p}_i + \boldsymbol{\epsilon}_i, \quad \boldsymbol{\epsilon}_i \sim \mathcal{N}_k(\boldsymbol{0}, \Sigma_\epsilon), \quad \boldsymbol{p}_i = \psi^{-1}(\boldsymbol{z}_i), \quad \boldsymbol{z}_i \sim \mathcal{N}_{m-1}(\boldsymbol{\mu}_p, \Sigma_p). \label{eq:full_model} \end{equation}


\subsection{Problem formulation}
\label{subsec:prob_form}
The parameters $(\mathbb{H}, \boldsymbol{\mu}_p, \Sigma_p, \Sigma_\epsilon)$ are estimated by maximising the marginal likelihood of the observed data $\{g_i\}_{i=1}^n$, obtained by integrating out the latent proportion vectors $\{\boldsymbol{p}_i\}$. Since this integral is intractable, we resort to the MCEM algorithm: in the E-step, the current conditional expectation of the complete-data log-likelihood is approximated via Monte Carlo; in the M-step, this approximation is maximised with respect to the parameters. Observing only $\{g_i\}_{i=1}^n \subset B^2(I)$ -- represented by $\{\boldsymbol{c}_i\}_{i=1}^n\subset \mathbb{R}^{k}$ -- , independent realizations of $G$, we seek to estimate via maximum likelihood the vertices $h_1,...,h_m\in B^2(I)$ through the coefficients' matrix $\mathbb H \in \mathbb{R}^{k\times m}$, the residuals' covariance operator $\Gamma_\epsilon\in \mathbb{R}^{k\times k}$ through $\Sigma_\epsilon$, the mean ilr-proportion $\boldsymbol \mu_p\in \mathbb{R}^{(m-1)}$ and its covariance $\Sigma_p\in \mathbb{R}^{(m-1)\times (m-1)}$.

Throughout this section, $f_{\boldsymbol{p}_i}$ denotes the model density of $\boldsymbol{p}_i$ induced by the ilr-normal assumption, and $f_{\boldsymbol{p}_i \mid \boldsymbol{c}_i}$ its conditional density given the observed coefficient vector $\boldsymbol{c}_i$. For given $\boldsymbol p_i$, $i\in\{1,...,n\}$, the log-likelihood 
would be: 

\begin{equation}
    l(\mathbb H,\Sigma_\epsilon|\boldsymbol{p}_1,...,\boldsymbol p_n)=\sum_{i=1}^n l_i(\mathbb H,\boldsymbol{p}_i,\Sigma_\epsilon)=
    -\frac{1}{2}\sum_{i=1}^n (\boldsymbol{c}_i-\mathbb H \boldsymbol p_{i})'\Sigma_\epsilon^{-1}(\boldsymbol{c}_i-\mathbb H \boldsymbol p_{i})-\frac{n}{2}\log(\det(\Sigma_\epsilon))+\text{const}.
\label{logL}
\end{equation}
However, in our scenario the proportions $\boldsymbol{p}_i$ are latent variables. Following the \emph{Monte Carlo Expectation--Maximization (MCEM)} framework, in each M-step we maximise the expected complete-data log-likelihood (the $Q$-function from the E-step) with respect to the parameters, namely:
\begin{equation}
l(\mathbb{H}, \boldsymbol{\mu}_p, \Sigma_p, \Sigma_\epsilon) = \sum_{i=1}^n \mathbb{E}_{\boldsymbol{p}_i| \boldsymbol{c}_i}\!\left[l_i(\mathbb{H}, \boldsymbol{p}_i, \Sigma_\epsilon) + \log f_{\boldsymbol{p}_i|\boldsymbol c_i}(\boldsymbol{p}_i)
\right],
    \label{ElogL}
\end{equation}
where the expectation is based on the conditional distribution of $\boldsymbol{p}_i| \boldsymbol{c}_i$ with the current values of the parameters. In each M-step, we aim to find the set of parameters $(\mathbb H, \boldsymbol \mu_p,\Sigma_p,\Sigma_\epsilon)$ which maximizes the expected log-likelihood in Equation \eqref{ElogL}.

The conditional expectation in \eqref{ElogL} is not analytically available. The $\boldsymbol{p}_i$  are obtained through a nonlinear transformation $\psi$ of a Gaussian random variable, but are not Gaussian themselves. As a consequence, the model density and the likelihood of $\boldsymbol{p}_i$ are not conjugate, and the resulting conditional distribution is not Gaussian. Nevertheless, expectations with respect to $\boldsymbol{p}_i$ can be approximated via Monte Carlo simulation. Specifically, at each iteration $t$ the current parameter estimates $(\mathbb{H}^{(t)}, \boldsymbol{\mu}_p^{(t)}, \Sigma_p^{(t)}, \Sigma_\epsilon^{(t)})$ are held fixed; the latent variables $\boldsymbol{p}_i$ are then sampled from the conditional $f_{\boldsymbol{p}_i| \boldsymbol{c}_i}$ (E-step), the expectation in~\eqref{ElogL} is approximated by Monte Carlo, and the parameters are updated by maximising this approximation (M-step). MCEM is well established in the literature \cite{bishop_book} and has already been employed in the Bayes Hilbert space setting \cite{steyer2023}.

The conditional distribution $\boldsymbol f_{\boldsymbol{p}_i| \boldsymbol{c}_i}$ of $\boldsymbol{p}_i$ is such that:
\begin{equation}
    f_{\boldsymbol{p}_i \mid \boldsymbol{c}_i}(\boldsymbol{p}_i) \propto \exp\!\left(-\frac{1}{2}(\psi(\boldsymbol{p}_i)-\boldsymbol{\mu}_p)'\Sigma_p^{-1}(\psi(\boldsymbol{p}_i)-\boldsymbol{\mu}_p) - \frac{1}{2}(\boldsymbol{c}_i - \mathbb{H}\boldsymbol{p}_i)'\Sigma_\epsilon^{-1}(\boldsymbol{c}_i - \mathbb{H}\boldsymbol{p}_i)\right).
\label{posterior_p}
\end{equation}

For the E-step, we can sample from the conditional distribution without explicitly computing the normalizing constant, using methods that range from importance sampling to more advanced algorithms such as Hamiltonian Monte Carlo (HMC). 
 Importance sampling is easy to implement and interpret, e.g. \cite{BECK2018523}, but poorly scalable as $m$ increases. Moreover, expectations are approximated by weighted averages, with weights given by the ratio between the target conditional density and the auxiliary density. Conversely, HMC is designed to efficiently explore high-dimensional conditional distributions \cite{Neal2011}, and in this case expectations are computed as simple sample averages, since all samples have unit weight and are already distributed according to the target conditional distribution. For this reason, we present both options. 

In the E-step, the expectation in~\eqref{ElogL} is approximated by drawing $B$ samples $\{\boldsymbol{p}_{b,i}\}_{b=1}^B$ from an auxiliary distribution $\pi$ and computing importance-weighted averages. Here $\boldsymbol{p}_{b,i}$ denotes the $b$-th sample drawn from $\pi$ for the $i$-th observation, with weight $w_{b,i}=\frac{\eta_{b,i}}{\sum_b \eta_{b,i}}$ and $\eta_{b,i}=\frac{f_{\boldsymbol{p}_i| \boldsymbol{c}_i}(\boldsymbol p_{b,i})}{\pi(\boldsymbol p_{b,i})}$. For HMC, $\pi = f_{\boldsymbol{p}_i \mid \boldsymbol{c}_i}$ and $w_{b,i} = 1/B$ for all $b$.

Notice that, together with the maximization of the expected log-likelihood with respect to the model parameters, we are also implicitly maximizing the expected log-density of $\boldsymbol{p}_i$. This term alone does not control the size of the induced convex hull $\mathcal{M}(\{h_j\}_{j=1}^m)$, since $\operatorname{tr}(\Sigma_p)$ may become arbitrarily small. In such a case, the proportions concentrate around the mean $\boldsymbol{\mu}_p$, while the locations of the vertices remain essentially unconstrained. To address this issue, and to improve the identifiability of the vertices -- see Section~\ref{identifiability} --, we introduce two additional penalization terms: one controlling the size of the convex hull by penalising small dispersion, and one regularising the shape of the estimated vertices when exact pure observations are not available. The resulting objective function to be \textit{maximised} is:
\begin{equation}
    \hat{l}(\mathbb{H}, \boldsymbol{\mu}_p, \Sigma_p, \Sigma_\epsilon)=\sum_{b=1}^B\sum_{i=1}^n w_{b,i}\bigl(l_i(\mathbb{H},\boldsymbol{p}_{b,i},\Sigma_\epsilon)+\log f_{\boldsymbol{p}_i|\boldsymbol{c}_i}(\boldsymbol{p}_{b,i})\bigr)-\lambda_{tr}\operatorname{tr}(\Sigma_p^{-1})-\lambda_{sm}\sum_j\left\|D^2[\operatorname{clr}(h_j)]\right\|^2_{L^2(I)},
    \label{hatl}
\end{equation}
where $D^2$ indicates the second derivative operator, and choice of $\lambda_{tr}$ and $\lambda_{sm}$ is discussed in Section \ref{sec:tuning}.

\subsection{Implementation}
In this section we describe the practical implementation of the E- and M-steps introduced in Section~\ref{subsec:prob_form}, showing that each parameter block admits a closed-form update within the coordinate-wise maximisation of the marginal likelihood. The resulting updates are reported below and form the building blocks of the iterative estimation procedure, as described in  Algorithm \ref{algo}.

\subsubsection*{Update of the vertices $\mathbb H$}

Fixing $\Sigma_\epsilon, \boldsymbol \mu_p,\Sigma_p$, the maximization of $\hat l$ has a closed form solution $\mathbb H^*$, that is:
\begin{equation}
    \operatorname{vec}(\mathbb H^*)=(A_r'\otimes I_{k\times k}+I_{m\times m}\otimes A_l)^{-1}\operatorname{vec}(C)
\label{optH}
\end{equation}
where $\otimes$ denotes the Kronecker product, $I_{k\times k}$ and $I_{m\times m}$ are, respectively, the identity matrix of dimension $k\times k$ and $m\times m$, $\operatorname{vec}(\cdot)$ is the operation that vectorises the matrices by column, and:

\begin{equation}
    \begin{cases}
        A_l=(2\lambda_{sm} \cdot \Sigma_\epsilon \cdot D_2) \\
        A_r=\sum_{b=1}^B\sum_{i=1}^n w_{b,i}\cdot(\boldsymbol p_{b,i}\boldsymbol p_{b,i}')\\
        C=\sum_{b=1}^B\sum_{i=1}^n w_{b,i}\cdot \boldsymbol{c}_i\boldsymbol p_{b,i}'
    \end{cases}
\end{equation}

where $D_2 \in \mathbb{R}^{k \times k}$ is the matrix representation of the second-derivative operator in the SFPC basis, i.e., $(D_2)_{kl} = \langle \phi_k'', \phi''_{l} \rangle_{L^2(I)}$. When $\lambda_{sm}=0$, the maximiser of $\hat{l}$ becomes $\mathbb H^* = C A_r^{-1}$. $\mathbb{H}^* = CA_r^{-1}$ can be interpreted as a weighted mean of the observations (encoded in $C$) corrected by $A_r^{-1}$, which adjusts for vertices' overlap. The diagonal entries of $A_r$ reflect average squared conditional proportions per vertex, while the off-diagonal entries capture their joint co-occurrence. Applying $A_r^{-1}$ thus mitigates vertex overlap. The derivation of~\eqref{optH} is provided in Appendix~\ref{appendix_proofs5}.

\subsubsection*{Update of the noise covariance $\Sigma_\epsilon$}
Conditional on $(\mathbb H, \boldsymbol \mu_p, \Sigma_p)$, 
the maximization of $\hat l$ with respect to $\Sigma_\epsilon$ 
reduces to the estimation of a weighted residual covariance matrix. 
The optimiser is given by
\begin{equation}
    \Sigma_\epsilon^*=\frac{1}{\sum_b\sum_iw_{b,i}}\sum_{b=1}^B\sum_{i=1}^nw_{b,i}(\boldsymbol{c}_i-\mathbb H \boldsymbol p_{b,i})(\boldsymbol{c}_i-\mathbb H \boldsymbol p_{b,i})'
    \label{opt_sigeps}
\end{equation}
that is, the conditional-weighted covariance of the reconstruction errors.
Hence, $\Sigma_\epsilon$ captures the dispersion of the data around the current mixture representation.
\subsubsection*{Update of the proportion's parameters}
Fixing $(\mathbb H, \Sigma_\epsilon, \Sigma_p)$, 
the update of $\boldsymbol \mu_p$ corresponds to a weighted average 
in the transformed coordinate system. 
Specifically,
\begin{equation}
    \boldsymbol \mu_p^*=\frac{1}{\sum_b\sum_i w_{b,i}}\sum_b\sum_iw_{b,i} \boldsymbol \psi(p_{b,i}),
    \label{opt_mup}
\end{equation}
that is, the conditional mean of the transformed proportions.
This expression highlights that the Gaussian assumption is enforced in the Euclidean space induced by $\psi$.

Finally, fixing $(\mathbb H, \Sigma_\epsilon, \boldsymbol \mu_p)$,
the maximization with respect to $\Sigma_p$ yields a regularised
weighted covariance estimator:
\begin{equation}
        \Sigma_p^*=\frac{1}{\sum_b\sum_i w_{b,i}}\bigg(\sum_b\sum_iw_{b,i}(\boldsymbol \psi(p_{b,i})-\boldsymbol \mu_p)(\boldsymbol \psi(p_{b,i})-\boldsymbol \mu_p)'+\lambda_{tr}I_{(m-1)\times (m-1)}\bigg),
\label{opt_sigp}
\end{equation}
where the trace-penalization term $\lambda_{tr} I$ ensures numerical stability
and prevents degeneracy of the covariance matrix.

The MCEM algorithm \cite{bishop_book}, with a coordinate-wise optimization in the M-step whose steps all have a closed-form solution (as shown in Equations \eqref{optH}, \eqref{opt_sigeps}, \eqref{opt_mup}, \eqref{opt_sigp}) is reported below. The algorithm proceeds as follows. At each iteration $\texttt{it}$, a simple random sample (SRS) -- i.e., drawn 
uniformly without replacement -- of $n_s = \lceil n\alpha_r \rceil$ observations is selected, where $\alpha_r \in (0,1]$ is the batch fraction. SRS ensures that each observation has equal probability of being selected, avoiding systematic bias 
in the parameter updates. Setting $\alpha_r = 1$ recovers the full-batch version of the algorithm. For each parameter block, conditional samples $\{\boldsymbol{p}_{b,i}\}$ are generated and the corresponding closed-form update is applied in sequence: first $\Sigma_\epsilon$, then $\mathbb{H}$, then $\boldsymbol{\mu}_p$, and finally $\Sigma_p$. Convergence is assessed via the relative change in $\mathbb{H}$; the algorithm stops when this falls below a tolerance threshold $\texttt{tol}$. The algorithm supports a \textit{stochastic-batch option}: rather than using all observations at each update step, a fresh SRS batch is drawn before each parameter update (lines 4, 6, 8, 10). This introduces additional stochasticity that can help escape local optima and reduces the computational cost per iteration, at the price of a noisier gradient signal. An alternative is to draw a single batch per full iteration. The tuning parameters $\lambda_{tr}$, $\lambda_{sm}$, and $m$ are discussed in Section~\ref{sec:tuning}. The batch fraction $\alpha_r$ should be set according to the redundancy of the dataset: when observations are spatially or otherwise correlated, $\alpha_r$ can be reduced substantially without loss of information -- in the case study of Section~\ref{sec:aviris} we set $\alpha_r = 0.2$ -- whereas for non-redundant datasets $\alpha_r = 1$ is recommended. For the convergence criterion, we set \texttt{tol}$= 0.05$ (i.e., a $5\%$ relative change threshold on $\mathbb{H}$) and \texttt{max\_it}$= 50$, which proved sufficient for convergence in all simulation and case study settings considered in this work. In absence of prior knowledge, we recommend to choose as initial values $\mu_p^{(0)}=0$, $\Sigma_p^{(0)}=I_{(m-1)\times(m-1)}$. The initialisation of $\mathbb{H}^{(0)}$ and $\Sigma_\epsilon^{(0)}$ ensures that the algorithm starts aligned with the directions of maximal variation in the data: the initial vertices coincide with the first $m$ simplicial principal components, while $\Sigma_\epsilon^{(0)}$ is calibrated to the variance explained by the first discarded component, as measured by the $(m+1)$-th eigenvalue.

\renewcommand{\algorithmicensure}{\textbf{Output:}}
\begin{algorithm}[H]
\footnotesize
\begin{algorithmic}[1]
\Require $\{\mathbf c_i\}_{i=1}^n$, $\Sigma_\epsilon^{(0)}$, $\mathbb H^{(0)}$, 
         $\mu_p^{(0)}$, $\Sigma_p^{(0)}$, $D_2$, 
         $\lambda_{tr}$, $\lambda_{sm}$, $\alpha_r$, $\texttt{max\_it}$, $\texttt{tol}$
\Ensure $(\mu_p, \Sigma_p, \mathbb H, \Sigma_\epsilon)$

\State $m \gets \text{ncol}(\mathbb H^{(0)})$
\State $n_s \gets \lceil n \alpha_r \rceil$

\For{$\texttt{it}=1$ to \texttt{max\_it}}
    
    \State SRS batch $\{\mathbf c_{i_s}\}_{i_s=1}^{n_s}$
    
    \State $\Sigma_\epsilon^{(\texttt{it})} \gets 
    \text{update}_{\Sigma_\epsilon}(\cdot)$ \hfill (Eq.~\eqref{opt_sigeps})
    
    \State SRS batch $\{\mathbf c_{i_s}\}_{i_s=1}^{n_s}$
    \State $\mathbb H^{(\texttt{it})} \gets 
    \text{update}_{\mathbb H}(\cdot)$ \hfill (Eq.~\eqref{optH})
    
    \State SRS batch $\{\mathbf c_{i_s}\}_{i_s=1}^{n_s}$
    \State $\mu_p^{(\texttt{it})} \gets 
    \text{update}_{\mu_p}(\cdot)$ \hfill (Eq.~\eqref{opt_mup})
    
    \State SRS batch $\{\mathbf c_{i_s}\}_{i_s=1}^{n_s}$
    \State $\Sigma_p^{(\texttt{it})} \gets 
    \text{update}_{\Sigma_p}(\cdot)$ \hfill (Eq.~\eqref{opt_sigp})
    
    \State $\Delta_H \gets \dfrac{\|\mathbb{H}^{(\texttt{it})} - \mathbb{H}^{(\texttt{it}-1)}\|_F}{\|\mathbb{H}^{(\texttt{it}-1)}\|_F + 10^{-8}}$
    
    \If{$\Delta_H < \texttt{tol}$}
        \State \textbf{break}
    \EndIf
\EndFor
\State \Return $\{\mu_p^{(\texttt{it})}, \Sigma_p^{(\texttt{it})}, \mathbb H^{(\texttt{it})}, \Sigma_{\varepsilon}^{(\texttt{it})}\}$
\end{algorithmic}

\caption{Algorithm for MCEM estimation -- with stochastic-batch option}
\label{algo}
\end{algorithm}

\subsection{Choice of the tuning parameters}
\label{sec:tuning}
In this section we discuss the choice of the penalty parameters $\lambda_{sm}, \lambda_{tr}$, and the number of vertices $m$.
\subsubsection*{Penalty parameters}
The parameter $\lambda_{sm}$ regulates the strength of the smoothing penalization. If the data are already smooth (e.g. in our simulations in Section \ref{sec:simulations}) or if they have to keep their natural irregularity (e.g. in the case study in Section \ref{sec:aviris}), $\lambda_{sm}$ can be set equal to 0. In general, $\lambda_{sm}>0$ can be selected on the basis of the reconstruction error of the mean density process, i.e.:
 \[
\text{MSE}_{CV}(\lambda_{sm}) = \frac{1}{|V|}\sum_{i \in V} \left\|\boldsymbol{c}_i - \mathbb{H}\boldsymbol{p}_i\right\|_2^2
\]

where $V \subset \{1, \ldots, n\}$ is a held-out validation set disjoint from the training set, and $\|\cdot\|_2$ denotes the Euclidean norm in $\mathbb{R}^k$. For real case studies, whose statistical analyses are commonly supervised by experts on the scientific subject of the data, one can also  choose $\lambda_{sm}$ based on the interpretability of the results. 

Unlike $\lambda_{sm}$, which enforces smoothness at the expense of fit to the data, the penalization of $\mathrm{tr}(\Sigma_p^{-1})$ can be interpreted as an effective increase, by a factor $\lambda_{tr}$, in the variance of each coefficient of $\psi(\boldsymbol{p})$ (see Equation \eqref{opt_sigp}). Empirically, the algorithm typically exhibits convergence in the $\|\cdot\|_2$ norm of the coefficients $\mathbb H$ to the same result, across a wide range of values of $\lambda_{tr}$, indicating a limited sensitivity of the solution to this regularization parameter when the distribution of proportions saturates near the vertices -- i.e. when 3P holds.

\subsubsection*{Number of vertices $m$}
The number of mixture components $m$ is treated as a model hyperparameter; its selection is validated empirically in Appendix~\ref{app:justify_m} for the simulation settings considered in Section~\ref{sec:simulations}. A classical approach consists in setting $m$ equal to the number of principal components selected via the elbow rule, or via a threshold on the cumulative sum of eigenvalues. Alternatively, one can fit the model for $m \in \{2,\ldots,m_{\max}\}$ and select $m$ via reconstruction error on a hold-out validation set (analogous to $\text{MSE}_{CV}$ in Section~\ref{sec:tuning}), an information criterion, or domain knowledge. The sensitivity of the results to $m$ should always be assessed by comparing solutions for adjacent values. The effect of $m$ on the estimation accuracy is further investigated in an additional simulation (Simulation~C, Appendix~\ref{app:n_vertices}).

\section{A simulation study: probability density unmixing}
\label{sec:simulations}
This section investigates the behavior of the proposed method under two relevant \textit{stress} scenarios, with increasing degree of difficulty. Each simulated dataset consists of $n = 500$ mixtures of $m$ density components (vertices) defined on the unit interval. Each component approximates a Beta distribution with varying shape parameters. Beta distributions are a well-known and flexible parametric family defined on the unit interval, making them a natural choice for simulating heterogeneous density shapes on $[0,1]$. However, as shown in Appendix~\ref{app:exponential}, exact Beta distributions are linearly dependent in the Bayes Hilbert space, which would violate the assumptions of Section~\ref{BHunmixing} (Definition~\ref{def_mix}). To circumvent this issue, each component is defined as a sample density obtained from a histogram generated from the corresponding Beta distribution, yielding vertices that are a.s. linearly independent. The densities are then transformed using the centered log-ratio (clr) transformation, which maps them from the Bayes space to an unconstrained Hilbert space. The resulting clr functions are represented through cubic compositional B-spline basis expansions with a second-derivative roughness penalty \cite{machalova2021}, as in \cite{eva2025}, yielding a finite-dimensional coefficient representation. 

Mixture observations are generated by simulating compositional weights in isometric log-ratio (ilr) coordinates from the multivariate Gaussian distribution
\[
\mathcal{N}_{m-1}(\mu_p, \Sigma_p),
\]
and subsequently transforming them back to the simplex. Functional mixtures are then obtained as linear combinations of the vertices in clr space. To mimic realistic data conditions, additive multivariate Gaussian noise,
\[
\varepsilon \sim \mathcal{N}_k(\boldsymbol{0}, \Sigma_\varepsilon),
\]
is introduced in the coefficient space.

This framework generates synthetic density-valued functional observations with a known underlying structure, enabling a systematic evaluation of estimation accuracy, robustness to noise, and recovery of both mixture components and mixing proportions.

The most challenging scenarios are expected to arise when: 
(i) the residual variance increases, i.e., $\operatorname{tr}(\Sigma_\varepsilon)$ becomes larger (Study A); 
(ii) the mean of the mixing distribution $\mu_p$ moves further away from the uniform composition $[0]_{S^{m}}$ towards a corner of the simplex (Study B).

Datasets are simulated under the following configuration. The number of mixture components is set to $m = 3$. The noise covariance matrix is specified as $\Sigma_\varepsilon = \sigma^2 I_{k \times k}$, where $\sigma^2 = 0.01 \cdot \operatorname{tr}(\widehat{\Sigma}_c)$, with $\widehat{\Sigma}_c$ the sample covariance of the clr-coefficient vectors $\{\boldsymbol{c}_i\}_{i=1}^n$, unless otherwise stated (see Study~A). The mean of the mixture coefficients is set to $\mu_p = \boldsymbol 0$ in Study A, while it varies in Study~B, in such a way to study its effect on the estimation.

The estimation phase is carried out according to Algorithm~\ref{algo}. To ensure comparability across scenarios, the following hyperparameters are kept constant in every simulation. The model distribution is initialized as follows: the mean vector is set to
\[
\mu_p^{(0)} = \boldsymbol{0},
\]
the covariance matrix of the mixture parameters to 
\[
\Sigma_p^{(0)} = I_{(m-1)\times(m-1)},
\]
and the noise covariance matrix is initialized as
\[
\Sigma_\varepsilon^{(0)} = I_{k\times k}.
\]
The matrix \(\mathbb H^{(0)}\) is initialized such that its columns correspond to the first \(m\) simplicial principal components of $\boldsymbol{c}_1,...,\boldsymbol{c}_n$ with $m$ set to $3$. The regularization parameters are fixed at
\[
(\lambda_{sm}, \lambda_{tr}) = (0,1),
\]
and the batch size is set to \(\alpha_r = 20\%\), meaning that each update is computed using 100 randomly sampled observations.

We claim that, apart from $m = 3$, the selected hyperparameters are intentionally not optimised for each individual realization of the synthetic datasets. This choice is motivated by the need to ensure comparability across the different scenarios and, at the same time, to avoid relying on a fully automatic tuning procedure, that would be computationally expensive in an extensive simulation study. The selection of $m = 3$ should not be regarded as an artificial advantage derived from the simulation design. Although $m = 3$ corresponds to the true underlying dimensionality used to generate the data, this choice is also strongly supported empirically (see Appendix~\ref{app:justify_m}).

\subsection{Study A: amplifying the dispersion from the convex hull}
\label{sec:simA}
In Study A, we investigate how the dispersion of the residuals affects the performance of our maximum-likelihood unmixing strategy. As $\operatorname{tr}(\Sigma_\epsilon)$ increases, the proportion of variance explained by the mixture correspondingly decreases. This makes the identification of the mixture more challenging, as the mixture structure becomes increasingly obscured by the overall variability of the data.

Figure~\ref{studyAB} (top row) illustrates a representative realization of Study A: the discrepancy between the true and estimated vertices is negligible when $\sigma^2$ equals $0.01$, corresponding to observations that are almost noiseless mixtures. However, the mismatch becomes more severe when $\sigma^2$ reaches $0.5$, that is, when the variability of the residual is comparable to half the variability of the vertices themselves. In this regime, the noise makes the estimated vertices also noisy. In situations like these the solutions can be:

\begin{itemize}
    \item Avoiding subsampling at each iteration, setting $\alpha_r=1$ -- see Algorithm \ref{algo}.
    \item Increase the smoothing penalization parameter $\lambda_{sm}$.
\end{itemize}

Figure~\ref{simstudyAB} (left) reports the $L^2$ error distribution across $50$ independent repetitions. The representative realization in Figure~\ref{studyAB} (top row) is selected as the one closest to the median error and should be interpreted as a typical outcome. Two features emerge: the median error remains roughly stable across values of $\sigma^2$, suggesting that the algorithm consistently identifies the correct vertex directions even under substantial noise; the interquartile range widens markedly as $\sigma^2$ increases, reflecting growing variability across repetitions. Two conditions are compared: in both cases, BH unmixing is applied, but the data are generated either from a Bayes Hilbert mixture (light blue boxplots) or from a linear mixture (violet boxplots). The latter case corresponds to a misspecified model, where the unmixing assumption does not match the data-generating process. As expected under misspecification, the linear mixture case yields systematically higher errors, confirming that the BH unmixing estimator is sensitive to the assumed mixing geometry.

\subsection{Study B: moving the mean proportion towards one vertex}
\label{sec:simB}
In Study B, we examine how the proximity to one particular vertex, and consequently the distance of $\mu_p$ from the uniform proportion $[0]_{S^m}$, influences the performance of our maximum-likelihood unmixing strategy. As $\mu_p$'s norm increases, at least one of the vertices becomes underrepresented in the dataset of the mixtures. This aspect can make the recognition of that vertex difficult. In particular, as the proportions relative to a vertex are distributed near 0, the influence of that vertex becomes almost null. In that specific case, the dataset will be approximately a mixture of $(m-1)$ vertices. For this simulation, we consider the following $\mu_p$ -- recalling that $m=3$: \[(0,0),(0,1/2),(0,1),(0,2)\] 
Figure~\ref{studyB_mup} illustrates the four choices of $\boldsymbol{\mu}_p$, represented as points in $S^3 \subset \mathbb{R}^3$. As $\|\boldsymbol{\mu}_p\|$ increases, the mean proportion moves away from the uniform composition towards one vertex. It is important to note that for all these simulations, the starting value for $\mu_p$ in the iterative procedure is always equal to $[0]_{S^m}$. As one can easily observe in the example in Figure \ref{studyAB} (bottom row), qualitatively, the estimated vertices progressively get worse from $[0]_{S^m}$ to $\psi^{-1}((0,1))\approx(0.18,0.18,0.64)\in S^3$. When the discrepancy between the proportion of one vertex is too high with respect to the ones of the other vertices, e.g. in our example for $\psi^{-1}((0,2))\approx(0.07,0.07,0.86)$, the two less represented vertices collapse into one unique vertex, becoming indistinguishable. 

Figure~\ref{simstudyAB} (right) summarises estimation accuracy across $50$ independent repetitions; the representative realization in Figure~\ref{studyAB} (bottom row) corresponds to the median error. Performance remains stable for $\boldsymbol{\mu}_p \in \{(0,0),(0,1/2),(0,1)\}$ and deteriorates sharply only at $\boldsymbol{\mu}_p = (0,2)$, where one vertex carries proportion approximately $0.86$ and the remaining two become nearly indistinguishable. This threshold behaviour suggests robustness of the estimation procedure to moderate imbalances in the mixing distribution, with breakdown occurring only when one vertex is so dominant that the others are effectively unobserved. As in Study~A, BH unmixing is applied to both conditions: data generated from a BH mixture (light blue) and data generated from a linear mixture (violet). The BH mixture case consistently yields lower errors throughout, confirming that the advantage of the correct model specification results also in this case -- regardless of $||\mu_p||$.

\begin{figure}[H]
    \centering
    \includegraphics[width=1\linewidth]{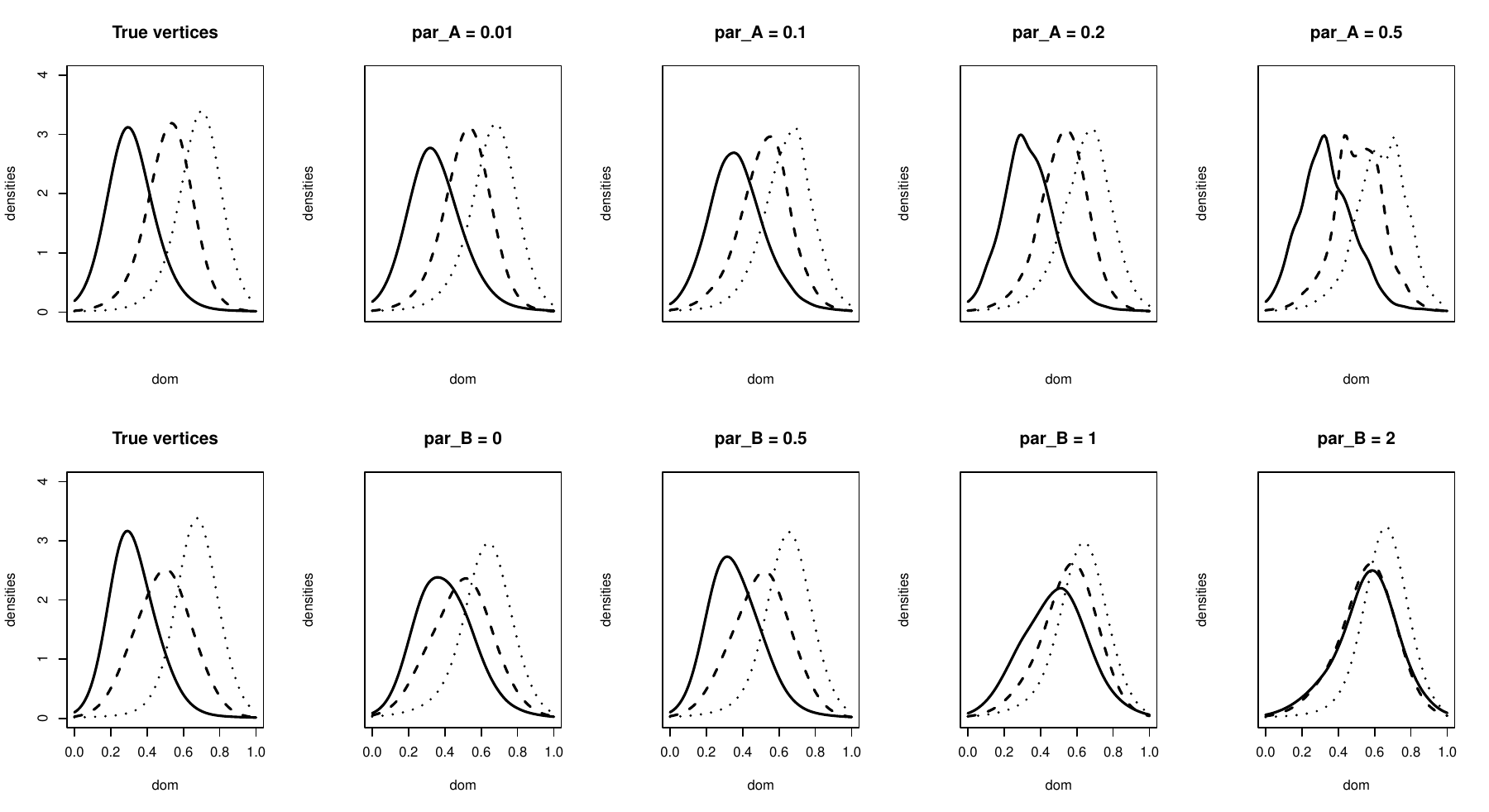}
\caption{Representative realizations of Study~A (top row) and Study~B (bottom row): true vertices (leftmost panel) and estimated vertices for increasing values of $\texttt{par\_a}=\sigma^2$ (Study~A) and $\texttt{par\_b}=\|\boldsymbol{\mu}_p\|$ (Study~B). Each realization is selected as the one closest to the median $L^2$ error over the $50$ repetitions.}
    \label{studyAB}
\end{figure}

\begin{figure}[H]
    \centering
    \includegraphics[width=1\linewidth]{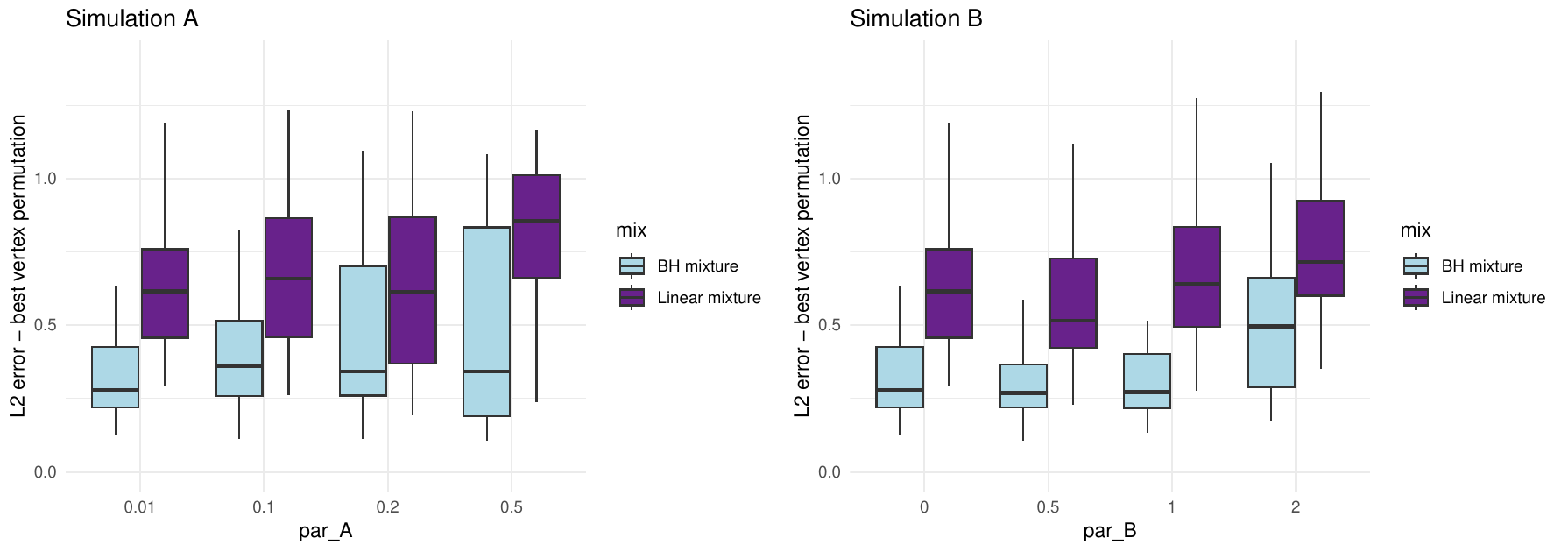}
\caption{Distribution of the $L^2$ estimation error -- defined as $\frac{1}{m}\min_\pi \sum_j \|h_j - \hat{h}_{\pi(j)}\|_{B^2}$, minimised over all vertex permutations $\pi$ -- across $50$ independent repetitions, as a function of $\sigma^2$ (Study~A, left) and $\|\boldsymbol{\mu}_p\|$ (Study~B, right). Light blue: data generated as Bayes Hilbert mixtures; violet: data generated as linear mixtures.}
    \label{simstudyAB}
\end{figure}

\section{A case study: AVIRIS Indian Pines hyperspectral unmixing}
\label{sec:aviris}
The AVIRIS Indian Pines dataset is a widely used benchmark in hyperspectral image analysis \cite{alfre2025}, providing a challenging setting for the evaluation of spectral analyses: from pixel classification  \cite{PhaneendraKumar2024} to unsupervised unmixing \cite{nascimento05}. Acquired by the Airborne Visible/Infrared Imaging Spectrometer (AVIRIS) in 1992 over agricultural areas in northwestern Indiana, the scene consists of 145 × 145 pixels with reflectance measurements across 224 spectral bands spanning the wavelength range from approximately $0.5\mu m$ to $2.5 \mu m$ \cite{PURR1947}. After removing bands affected by atmospheric absorption and noise, the dataset exhibits high spectral redundancy, subtle class variability, and significant spectral mixing due to limited spatial resolution and heterogeneous land cover. Moreover, the Indian Pines site is predominantly characterised by the coexistence of vegetation, water, and soil, with only a limited presence of artificial structures such as streets and buildings -- see Figure \ref{1993_2003}. As a consequence, each pixel is likely to simultaneously represent multiple land-cover types. The combined effect of low spatial resolution -- resulting 
in spatial blurring -- and the intrinsic heterogeneity of the scene makes Indian Pines particularly suitable for testing advanced unmixing models that go beyond linear and pointwise representations of spectra. In this setting, Bayes Hilbert mixtures are a natural modelling choice: rather than describing 
each pixel as a discrete combination of pure endmembers, they yield a single coherent distribution that integrates the spectral characteristics of all contributing land-cover types, consistently with the fuzzy nature of mixed pixels.

In this work, we interpret each hyperspectral pixel not as a finite-dimensional vector but as a functional compositional object, viewed as a distribution over the continuous wavelength domain. This perspective naturally leads to an embedding of hyperspectral data into the Bayes Hilbert space, despite it being originally designed for analysing probability density functions. This approach enables us to focus on the \textit{shape} of the spectra, rather than on their amplitude. Indeed, it is widely recognised in the literature that amplitude has no physical meaning in hyperspectral analysis.

Within this functional-compositional framework, spectral unmixing is reinterpreted as a density unmixing problem. We model each observed normalised hyperspectrum as a BH mixture -- as mixing is thought to occur within pixels, i.e. fuzzy-like, as discussed in Section \ref{BHunmixing} --  of latent \textit{endmember} densities  (vertices) defined over the wavelength domain. Figure \ref{ip_vertices} on the left shows the dataset of normalised hyperspectra. Applying the simplicial FPCA \cite{Hron16}, we select $m=4$. We set the starting values for $\{h_j\}_{j=1}^m$ as the first $m$ principal components, and for $\Sigma_\epsilon$ as $\rho_{m+1}I_{k\times k}$, where $\rho_{m+1}$ is the $(m+1)$-th eigenvalue, i.e. the variance explained by the first discarded principal component. The first $4$ principal components, obtained without imposing smoothness so as to preserve the roughness typical of hyperspectral data, and the variance explained are reported in Figure~\ref{app:pca_var} in Appendix~\ref{app:aviris_figures}.

The choice of $m=4$ is guided by Figure~\ref{app:pca_var} (center), i.e. the plot of variance explained varying the number of retained principal components. This plot displays an elbow in $m=4$ -- retaining $96\%$ of the total variance.


Selecting $m=4$, we discuss the result of our MCEM unmixing, employed to estimate the tuple of vertices $\{h_1, h_2, h_3, h_4\}$, the mean ilr-proportion $\mu_p$ and covariance $\Sigma_p$, and residual coefficients' covariance $\Sigma_\epsilon$. In Figure \ref{ip_vertices} we show the data (left) compared to the estimated vertices (right). The vertices seem reasonable convex generators of the data. In order to interpret the results, we plot, for each pixel of the Indian Pines site, the \textit{maximum a posteriori} estimate of each component of the proportion $p_j$, for every $j \in \{1,...,m\}$, obtaining $m$ maps -- see Figure \ref{figs:bhm_casestudy}. 

Comparing the results obtained with historical imagery from 1993's Google Earth, with the ten years later map -- Figure \ref{1993_2003} -- and keeping as a reference the experts' manually labeled landcover map -- Figure \ref{manual} -- we can state that the proposed spectral unmixing strategy is capable of accurately distinguishing the following land-cover classes: 
\(h_1\) – cultivated lands of type 1, 
\(h_2\) – cultivated lands of type 2, 
\(h_3\) – impervious surfaces, including streets and building roofs,
\(h_4\) – trees. 
By considering the sum of the proportions of $h_1$ and $h_2$, we can easily recognise the profile of the cultivated fields in 1993. The difference between $h_1$ and $h_2$ is unknown a priori. However, we can observe that $h_1$ has a stronger decay in the band $[0.55 \mu m, 1 \mu m]$, classified as the \textit{Visible to Near Infrared} (VNIR). In this band, water reacting with chlorophyll -- the green pigment, present in all green plants and in cyanobacteria, responsible for the absorption of light to provide energy for photosynthesis -- typically present a higher reflectance \cite{nasa_site}. Indeed, this band is observed by scientists in order to assess if a plant is healthy -- high reflectance -- or not. Thus, we can suggest that $h_1$ represents fields that are less healthy than the ones with a high proportion of $h_2$. Analogously, the difference could also be due to a higher presence of water in $h_2$.

By comparison, we report in Appendix~\ref{app:aviris_figures} (Figure~\ref{app:vca_casestudy}) the result of Vertex Component Analysis (VCA, \cite{Nascimento2005}) on AVIRIS Indian Pines, obtained on the normalized hyperspectra, trained on $25\%$ of the pixels -- to reduce redundancy induced by spatial correlation --, fixing $m=4$ to make the result comparable to ours. VCA is a statistical method for \textit{linear} unmixing that requires a pure-pixel assumption. It extracts the endmembers $(\nu_j)_{j=1}^m$ among the statistical units, and computes the pixels' abundances vectors $\boldsymbol{a}=(a_j)_{j=1}^m$, which, in contrast to proportions, are only constrained to be positive, not to sum to 1. Consequently, two endmembers can have a high abundance on the same pixel. Employing as a reference Figures \ref{manual_summarised} and \ref{1993_2003} shows that $\nu_3$ highlights the presence of trees, similarly to $h_4$, $\nu_2$ the grass-pastures (see Figure \ref{manual}), while $\nu_1,\nu_4$, though spatially coherent, are of difficult interpretability both in terms of manually labeled map and of Google Earth's photographies (Figure \ref{1993_2003}).

The interpretable findings and the coherence of our results highlight the potential of functional-compositional representations for hyperspectral unmixing, particularly in scenarios where spectral variability and spatial blurring challenge traditional linear models. Indeed, as illustrated in Section~\ref{BHunmixing}, the Bayes Hilbert mixture naturally models spatially blurred observations, where distinct endmembers are smoothly integrated within a single pixel.

\begin{figure}[H]
    \centering
 \includegraphics[width=1\linewidth]{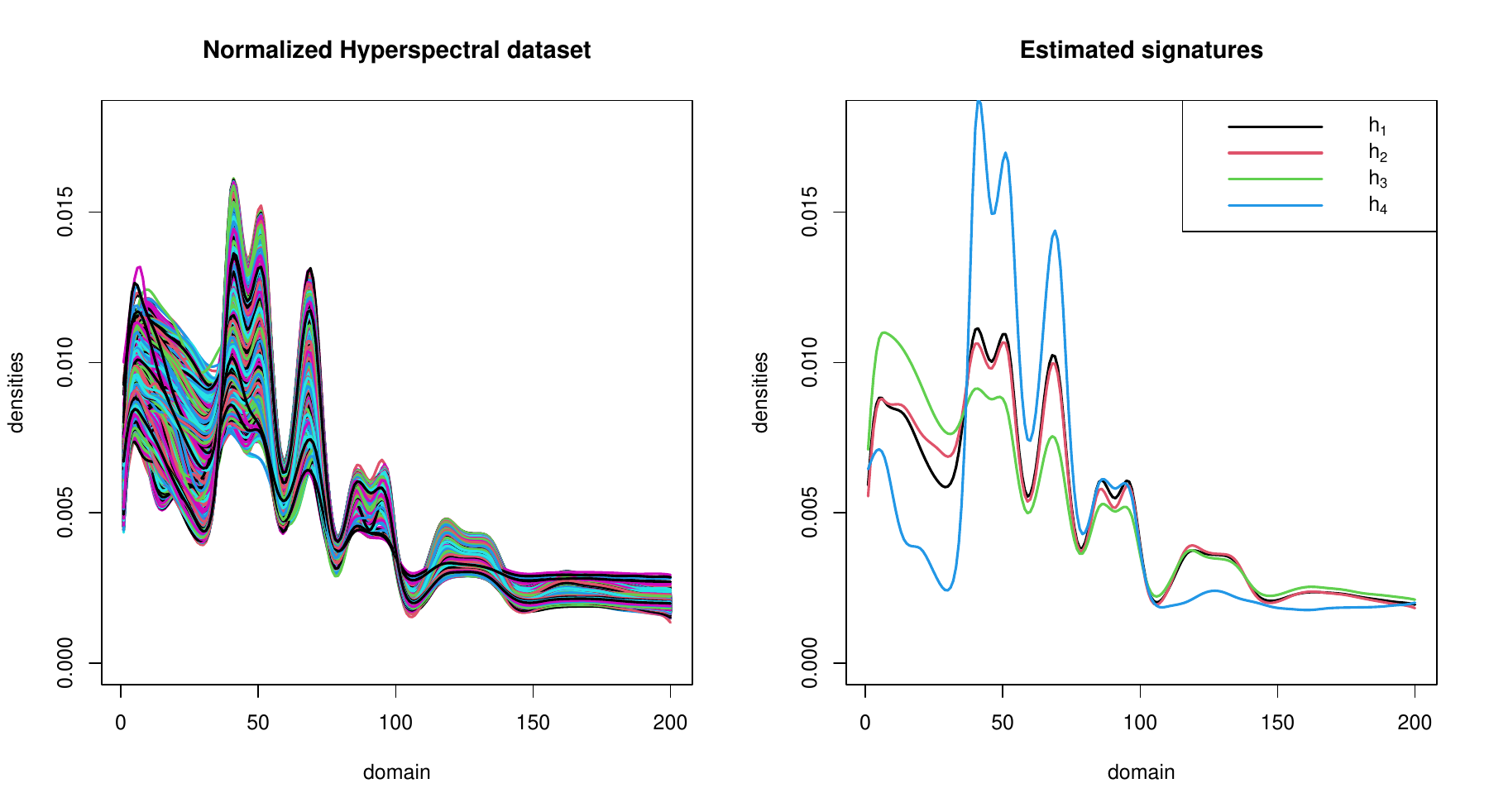}
    \caption{On the left the normalised hyperspectral mixtures, to be unmixed. On the right the $m=4$ estimated vertices.}
    \label{ip_vertices}
\end{figure}

\begin{figure}[H]
    \centering
    \includegraphics[width=1\linewidth]{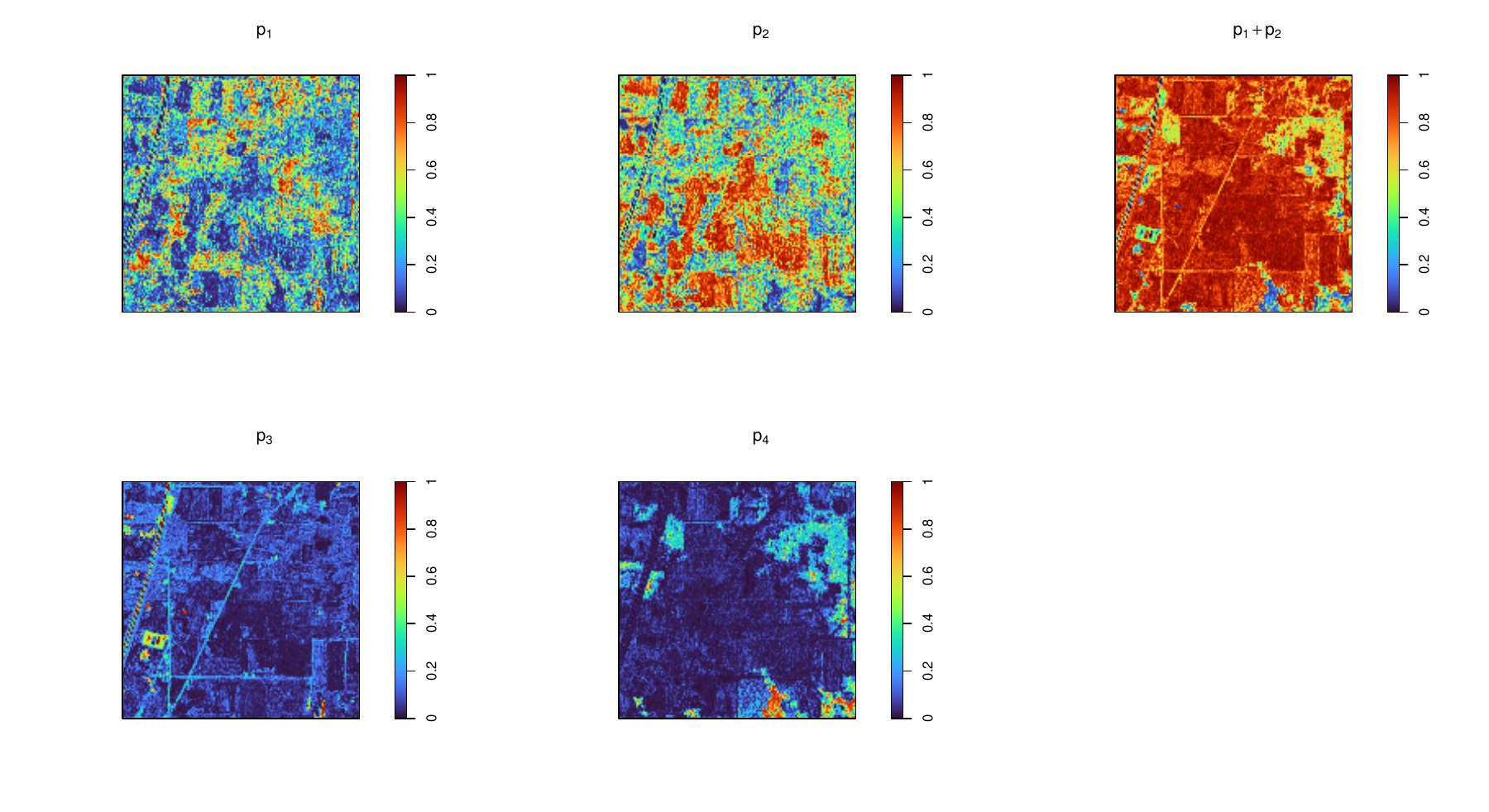}
    \caption{Unmixing in Bayes Hilbert space. The five panels show the estimated proportions $p_1, \dots, p_4 \in (0,1)$ as well as $p_1+p_2$. 
    }
    \label{figs:bhm_casestudy}
\end{figure}

\section*{Conclusion}

In this work, we introduced a framework for the analysis and unmixing of random density mixtures in the Bayes Hilbert space $B^2(I)$. We studied the identifiability of a statistically space-efficient representation of a random mixture, defined as the convex combination of vertices in $B^2(I)$, with vertex proportions following a distribution that maximizes the covered region of the simplex.
General identifiability results for mixtures in Hilbert spaces were established and subsequently applied to the Bayes Hilbert space setting.

Based on these results, we proposed a penalized maximum likelihood approach for the unmixing of Bayes Hilbert mixtures aimed at recovering the statistically space-efficient representation. The resulting optimization problem can be implemented through a coordinate-wise maximization algorithm.

The proposed methodology was illustrated through a hyperspectral data application, where observations can be naturally embedded in the Bayes Hilbert space and analyzed in terms of distributional shape rather than amplitude. A complementary simulation study based on noisy mixtures of Beta distributions further demonstrated the interpretability and practical performance of the method, particularly in the presence of fuzzy-type mixtures.

The performance of the approach may deteriorate when observations lie far from the center of the convex hull of the vertices, are very noisy or when the mixture type is misspecified, for instance when the underlying mixture is linear. An interesting direction for future research is to relax the assumption of strictly Bayes--Hilbert or linear mixtures and investigate which $\alpha$-mixture within the $\alpha$-mixture family \cite{AsadiEbrahimiSoofi2019, Shojaee2021, ShojaeeAsadiFinkelstein2022, ShojaeeMomeni2023} provides the most appropriate representation of a generic random mixture of densities. Concerning the model for the proportions, the $\operatorname*{ilr}$ framework may become restrictive under the pure-pixel assumption, since logarithmic transformations are undefined for zero components; in such cases, the discrete $\alpha$-transformations \cite{Tsagris2016, ClarottoAllardMenafoglio2022} may provide a more suitable alternative, as they admit pure proportions and remain well-defined on the boundary of the simplex. In our model, we assume that the proportion vector $\boldsymbol{p}$ follows an $\mathrm{ilr}$-normal distribution, that is, its isometric log-ratio ($\mathrm{ilr}$) transformation is multivariate normally distributed. This choice yields a unimodal distribution for $\boldsymbol{p}$. When the observed realisations of $G$ exhibit multimodality -- suggesting the presence of distinct subpopulations -- a more flexible model for the distribution of $\boldsymbol{p}$ may be considered \cite{ComasCufi2016}. 

Finally, one could consider modelling correlation --temporal, spatial, or spatio-temporal, depending on the application -- in the proportions. Notably, in the hyperspectral case study of Section~\ref{sec:aviris}, 
spatially contiguous proportion maps are obtained even without explicitly accounting for 
spatial dependence (see Figure~\ref{figs:bhm_casestudy}). 

We believe that the Bayes Hilbert space setting provides a principled and interpretable 
foundation for fuzzy-like mixture modelling of distributional data, and that the extensions outlined above represent promising directions for broadening its scope and applicability.








\section*{Acknowledgments}
GP and AM acknowledge the support provided by the European Commission under the “HORIZON-CL4-2021-\\ DIGITALEMERGING-01 project BioProS - Biointelligent Production Sensor to Measure Viral Activity” (grant agreement no. 101070120), 2022-2026”. GP and AM acknowledge the initiative “Dipartimento di Eccellenza 2023–2027”, MUR, Italy, Dipartimento di Matematica, Politecnico di Milano. Funded by the Deutsche Forschungsgemeinschaft (DFG, German Research Foundation) - Project number 513634041. GP acknowledges Maarten Jung, Manuel Pfeuffer, Johannes M. Feeser and Dr. Georg Keilbar for the fruitful scientific exchanges on functional and density data analysis, and Alfredo Gimenez Zapiola for sharing his expertise on hyperspectral unmixing.

\bibliography{els-cas-templates/biblio_check}
\appendix
\section{BH mixtures in exponential families}
\label{app:exponential}
As an illustrative example, we briefly discuss the special case of distributions belonging to the \textit{exponential families}. We already know that the set of the exponential family of distributions is an affine subspace of $B^2(I)$ \cite{VanDenBoogaart2014}. Here, we want to prove that the Bayes Hilbert mixture operation preserves the \textit{exponential distribution families}: if all elements in the convex hull of \(\{h_j\}\subset B^2\) belong to the same exponential family, then the Bayes Hilbert mixture also belongs to that same family, yielding a density characterized by updated parameters.

\begin{proposition}
Let $\mathcal F$ be the set of all densities obtained by truncating a fixed exponential family to a compact support $I' \subset I$.
Then $\mathcal F$ is a convex subset of $B^2(I)$.
\label{expo_fam_prop}
\end{proposition}
Proof:
Convex combinations in $B^2(I)$ are defined through the perturbation and powering operations
$\oplus$ and $\odot$. Let $f_1,f_2 \in \mathcal F$ have the same compact support $I' \subset I$, and let $p \in [0,1]$.
Assume
\[
\begin{cases}
f_1(x)=h(x)\exp\!\big(\eta(\theta_1)\cdot T(x)-A(\theta_1)\big)\,\mathbbm{1}_{I'}(x),\\[4pt]
f_2(x)=h(x)\exp\!\big(\eta(\theta_2)\cdot T(x)-A(\theta_2)\big)\,\mathbbm{1}_{I'}(x).
\end{cases}
\]

where $\mathbbm{1}_I(x)$ denotes the indicator function of the set $I$ at $x$. For $x \in I'$, the convex combination in $B^2(I)$ is
\[
(p \odot f_1)\oplus\big((1-p)\odot f_2\big),
\]
By the algebraic properties of $B^2(I)$, this is equivalent to
\begin{equation} \begin{split} &[p\odot (h(x)\cdot \exp(\eta(\theta_1)\cdot T(x)-A(\theta_1)))]\\ &\oplus[(1-p)\odot h(x)\cdot \exp(\eta(\theta_2)\cdot T(x)-A(\theta_2))]\\ =_{B^2}& h(x)^{(p+(1-p))}\exp((p\eta(\theta_1)+(1-p)\eta(\theta_2))T(x)\\ -&(pA(\theta_1)+(1-p)A(\theta_2))]\\ 
=_{B^2}&h(x) \exp(\eta(\theta_{12})T(x))\cdot cost\\
=_{B^2}& h(x) \exp(\eta(\theta_{12})T(x)-A(\theta_{12}))
 \end{split} \label{expo_fam} \end{equation}

By assuming the space of parameters $\Theta$ being convex, the existence of $\theta_{12}\in\Theta$ follows from the continuity of the
natural parameter map $\eta$ (by the intermediate value theorem). Hence, the convex combination belongs to the same exponential family truncated to $I'$,
and therefore lies in $\mathcal F$.

\qed

Hence, differently from the probabilistic mixture,  which would generally exhibit multimodality, the BH mixture produces a distribution of the same exponential family but with modified parameters.

\section{Propositions for identifiability of the minimal representation}
\label{app:identifiability}
This appendix shows and discusses the details regarding the results, presented in Section~\ref{identifiability}, on the identifiability of the representation of a random mixture $G$ through a set of linearly independent vertices $\{h_j\}_{j=1}^m$ and the distribution of the random proportion. Here, we assume that $G$ is a \textit{pure} random mixture, meaning that none of the possible realisations of $G$ are allowed to fall outside the convex hull of a set of $m$ vertices, where $m$ is fixed -- known or estimated. All the results in this Section are valid for any Hilbert space, although here presented in $B^2(I)$. Moreover, as in Section~\ref{identifiability}, $\operatorname{supp}(G)$ is compact, $G$ is non-degenerate, and the vertices $\{h_j\}_{j=1}^m$ are restricted to a fixed finite-dimensional subspace $\mathcal{H} \subset B^2(I)$; the latter restriction is needed in particular for the existence of a minimal representation (Proposition~\ref{prop:min_exist}). This is consistent with practical implementations of the model (see Section~\ref{pmlu}), where vertices are represented through a finite basis expansion (e.g., simplicial functional principal components), and is a standard restriction in the Bayes Hilbert space literature when existence results are sought (see e.g. \cite{Hron16, steyer2023} for an analogous distinction between infinite-dimensional and finite-dimensional recovery theorems).
 
 As a first step, in \ref{part_ident} we prove that, if $G$ is a random mixture, any set of vertices $\{h_j\}_{j=1}^m$ can be associated, according to $G$, with a unique random proportion $\boldsymbol{p}:\Omega\to S^m$. We will denote by $\mathfrak{D}(S^m)$ the set of all probability distributions on $S^m$, and write $\boldsymbol{p}\sim D\in \mathfrak{D}(S^m)$ to denote the distribution of $\boldsymbol{p}$.
 
\subsection{Partial identifiability}
\label{part_ident}
 
\begin{figure}[!htbp]
        \centering
        \includegraphics[width=\textwidth]{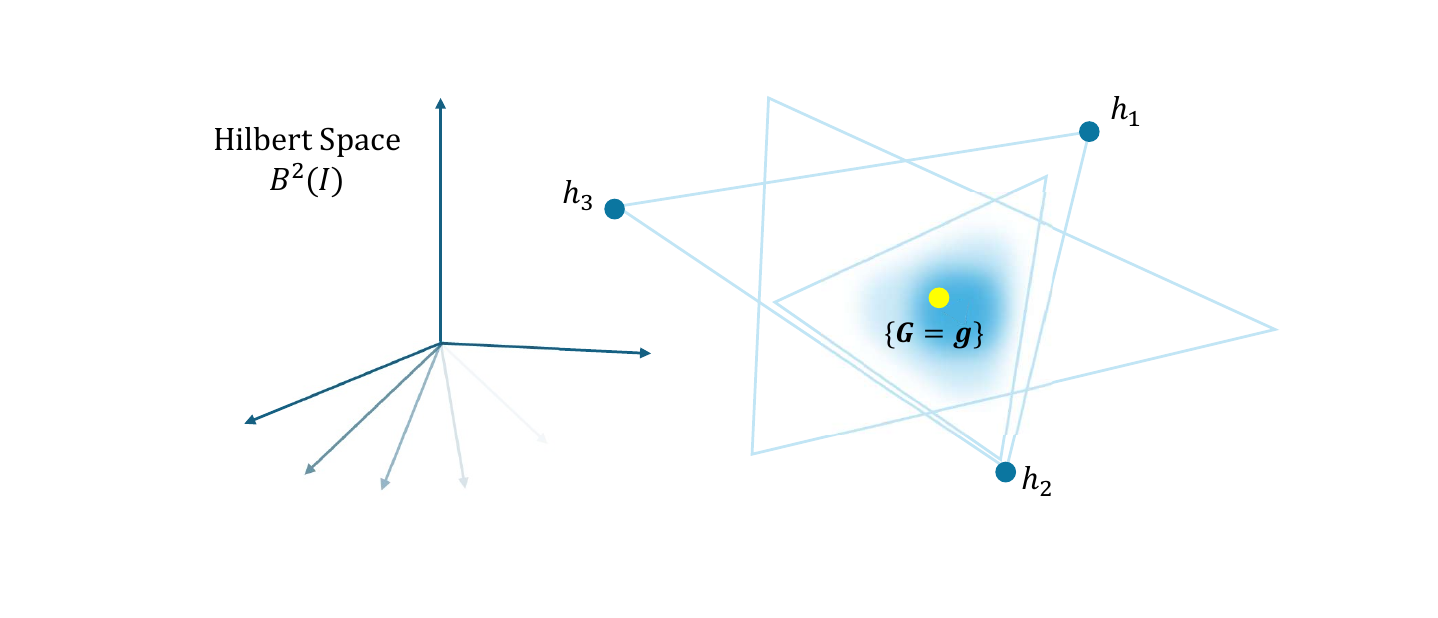}
\caption{Schematic representation of a mixture $G$ lying perfectly in $\mathcal{M}(h_1,h_2,h_3)$. The points $h_1, h_2, h_3$ are represented as points in the Bayes Hilbert space, and the triangle with vertices $h_1, h_2, h_3$ illustrates a convex hull. The blue cloud represents the distribution of $G$, with darker blue indicating higher probability mass, while the yellow point corresponds to a realization of $G$. Both the blue region and the realization of $G$ lie within the depicted convex hull with probability 1. The other triangles, with unmarked vertices, represent alternative convex hulls in which $G$ also lies with probability 1, but corresponding to a different distribution of the proportions. In this sense, the representation of $G$ as a random mixture is not unique.}

                \label{intro_mix}
    \end{figure}

Let us assume that $G$ is a random mixture of $\{h_j\}_{j=1}^m$, according to Definition \ref{def_mix}. Then, by construction, for any $\omega \in \Omega$ we may find $\boldsymbol p(\omega)$ such that
\[ 
G(\omega)= \bigoplus_{j=1}^m p_j(\omega)\odot h_j.\]
Hence, the random mixture $G$ identifies at least one random variable $\boldsymbol{p}:\Omega\to S^m$. Proposition \ref{prop:unique_p} states that the distribution of $\boldsymbol{p}$, given the distribution of $G$ and $\{h_j\}_{j=1}^m$, is unique.
\begin{proposition}[Identifiability in distribution of $\boldsymbol p|\{h_j\}_j$]
    Let $G$ be a random mixture of $\{h_j\}_{j=1}^m$, with $h_j$ non-random , linearly independent elements in $B^2(I)$, $j=1,...,m$, i.e. for $\omega\in \Omega$: 
    \[G(\omega)=\bigoplus_{j=1}^m p_j(\omega)\odot h_j\] for some  $\boldsymbol{p}\sim D\in \mathfrak{D}(S^m)$. If it also holds: 
    \[G\operatorname*{=}^{d}\bigoplus_{j=1}^m \bar p_j\odot h_j\] for some random vector $\boldsymbol{\bar p}$, then $\boldsymbol{p}\operatorname*{=}^d \boldsymbol{\bar p}$.
    \label{prop:unique_p}
\end{proposition}
 
\textbf{Proof of Proposition \ref{prop:unique_p}}: By contradiction, let us assume that $\boldsymbol{p}\operatorname*{\not=}^d \boldsymbol{\bar p}$. Then, there exists a Borel set $A \subseteq S^m$ such that $\mathbb{P}(\boldsymbol{p}\in A)<\mathbb{P}(\boldsymbol{\bar p}\in A)$. Let us define $A_G:=\bigg\{g=\bigoplus_{j=1}^m p_j\odot h_j\bigg| \boldsymbol{p}\in A \bigg\}$ and $\forall \omega\in\Omega$, $\bar G(\omega):= \bigoplus_{j=1}^m \bar p_j(\omega)\odot h_j$. The map $\boldsymbol p\mapsto\bigoplus_j p_j\odot h_j$ is injective by linear independence of the vertices, and continuous, so that $A_G$ is a Borel set and $\{G\in A_G\}=\{\boldsymbol p\in A\}$): \[
\mathbb{P}(G\in A_G)=\mathbb{P}(\boldsymbol{p}\in A)<\mathbb{P}(\boldsymbol{\bar p}\in A)=\mathbb{P}(\bar G \in A_G)
\]
And this contradicts the assumption that $G\operatorname*{=}^d \bar{G}$.
\qed 
 
\subsection{The vertices set}
Not every $m$-tuple of vertices in $B^2(I)$ generates a valid convex hull for $G$ (see Figure \ref{image_vgm}). To represent the distribution of $G$ as a random mixture, we therefore restrict attention to those $m$-tuples whose elements are linearly independent and for which $G$ lies in their convex hull with probability 1 -- that is, $supp(G) \subseteq \mathcal{M}(\{h_j\}_{j=1}^m)$. The collection of such $m$-tuples defines the subset $\mathcal{H}^m_G \subset B^2(I)\times ... \times B^2(I)= B^{2,m}(I)$ -- Definition \ref{def:H^m_g}. Furthermore, we show that , under non-degeneracy, $\mathcal{H}^m_G$ is closed in $B^{2,m}(I)$.

\label{vertices_set}
\begin{definition}
\label{def:H^m_g}
    We call $\mathcal{H}^m_G$ the subset of $\mathcal{H}^m$ defined as:
    \[
    \mathcal{H}^m_G:=\{(h_1,...,h_m)\in \mathcal{H}^m: \dim(\operatorname{span}_{B^2(I)}(h_1,...,h_m))=m,\]
    \[\mathbb{P}\left(G\in {\mathcal{M}(\{h_j\}_{j=1}^m)}\right)=1\}
    \]
\end{definition}
where $\operatorname{span}_{B^2(I)}$ is the span in the Bayes Hilbert space sense.
\begin{figure}[!htbp]
    \centering
    \includegraphics[scale=0.5]{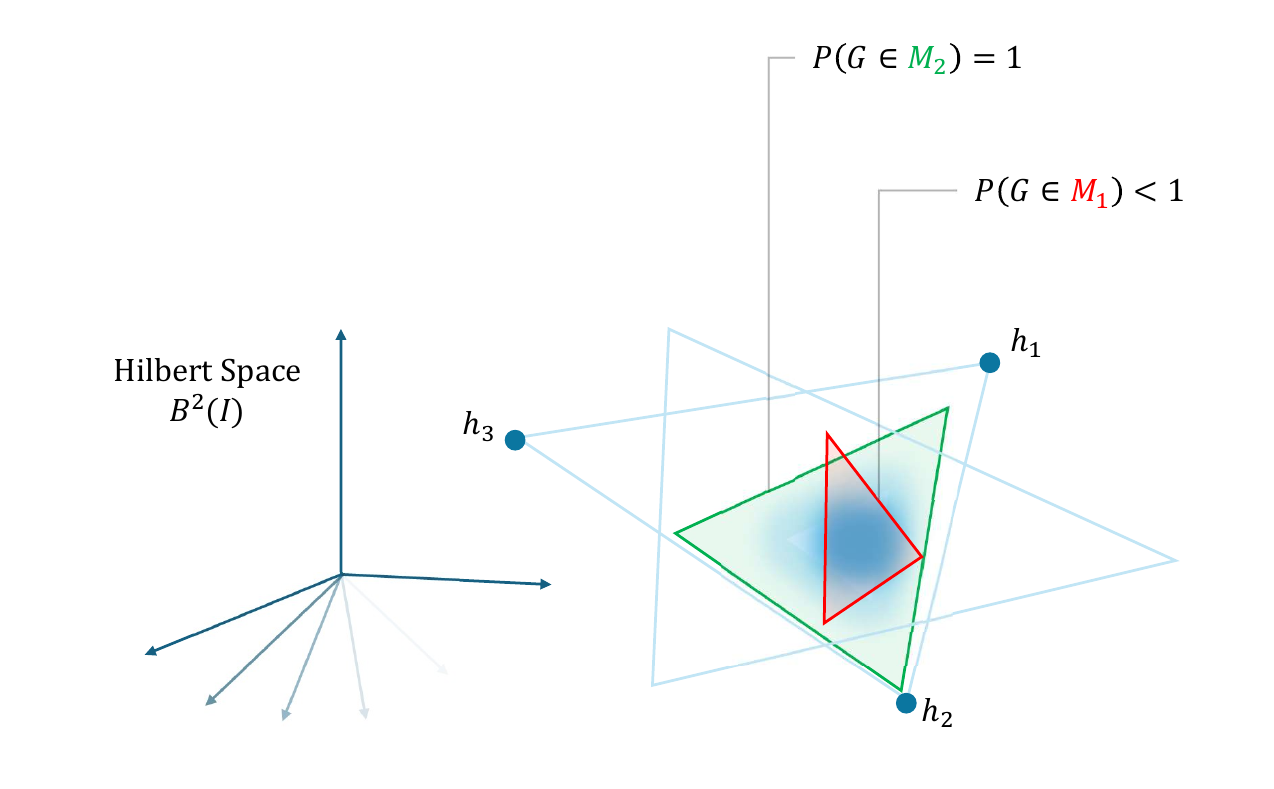}
    \caption{All triangles represent generic convex hulls. The green convex hull $(M_2)$ is an example of element of $\mathcal{H}^m_G$, as $supp(G)$ is contained in it. Contrarily, the red one $(M_1)$ does not entirely contain $supp(G)$, consequently it does not belong to $\mathcal{H}^m_G$.}
    \label{image_vgm}
\end{figure}

Now, let us provide $B^{2,m}(I)$ with the scalar product $\langle h,h'\rangle_{B^{2,m}(I)}=\sum_{j=1}^m \langle h_j,h'_j\rangle_{B^2(I)}$ and the induced norm $\|h\|_{B^{2,m}(I)}=\big(\sum_{j=1}^m \|h_j\|^2_{B^2(I)}\big)^{1/2}$. Since $B^2(I)$ is a Hilbert space, so is $(B^{2,m}(I),\|\cdot\|_{B^{2,m}(I)})$. We endow $\mathcal{H}^m\subset B^{2,m}(I)$, where $\mathcal{H}$ is a finite-dimensional subspace of $B^2(I)$, with the restriction of this norm, denoted by $\|\cdot\|_{\mathcal{H}^m}$. Being a finite-dimensional subspace of a normed space, $\mathcal{H}^m$ is closed in $B^{2,m}(I)$. This allows us to obtain the result in Proposition \ref{prop:close_H_G^m}, whose role will be clarified in Proposition \ref{prop:min_exist}.
 
\begin{proposition}
\label{prop:close_H_G^m}
If $G$ is non-degenerate, the set $\mathcal H_G^m$ is closed in $\mathcal{H}^m$.
\end{proposition}
\textbf{Proof of Proposition \ref{prop:close_H_G^m}}: Consider a sequence of $m$-tuples of vertices $\{h_j^{(\iota)}\}_{j=1}^m\in\mathcal H^m_G$, $\iota\in\mathbb{N}$, converging to $\{h_j^*\}_{j=1}^m$. From the closure of $\mathcal{H}^m$, we know that $\{h_j^*\}_{j=1}^m\in \mathcal{H}^m$. We have to prove that (a) $\mathbb{P}(G\in \mathcal{M}^*)=1$, where $\mathcal{M}^*:=\mathcal{M}(\{h^*_j\}_j)$, and (b) $h^*_1,\dots,h^*_m$ are linearly independent.
 
(a) Consider a sequence of convex hulls $\mathcal{M}_\iota:=\mathcal{M}(\{h_j^{(\iota)}\}_j)_\iota\subset\mathcal{H}_G^m\subset \mathcal{H}^m$, converging to a set $\mathcal{M}^*$, that belongs to $\mathcal{H}^m$ as it is closed. If the sequence is increasing, namely $\mathcal{M}_{\iota}\subset \mathcal{M}_{\iota+1}$, then, straightforwardly $P(G\in \mathcal{M}^*)\geq P(G\in \mathcal{M}_\iota)=1$. If it is decreasing, namely $\mathcal{M}_{\iota+1}\subset \mathcal{M}_{\iota}$, then, being $P(G\in {\mathcal{M}}_\iota)=1$ for any $\iota\in\mathbb N_0$, and following the properties of the probability as a measure,
\[
1=\lim_{\iota\to\infty} P(G\in {\mathcal{M}}_\iota)=P(\bigcap_{\iota=1}^\infty\{G\in {\mathcal{M}}_\iota\})= P(G\in \mathcal{M}^*)
\]
This implies that $\mathcal{M}^*\in \mathcal{H}_G^m$. Lastly, if the sequence is neither decreasing nor increasing, we can use the fact that
\[
\bigcap_{\iota=1}^\infty\{G\in {\mathcal{M}}_\iota\} \subset\bigcup_{n=1}^{\infty}\bigcap_{\iota=n}^\infty\{G\in {\mathcal{M}}_\iota\}=\lim_{\iota\to \infty} \mathcal{M}_\iota=\mathcal{M}^*.
\]
Thus we can conclude that
 \[
1=P(\bigcap_{\iota=1}^\infty\{G\in {\mathcal{M}}_\iota\})\subset P(G\in \mathcal{M}^*)\leq 1
 \]
(b) By non-degeneracy, $\operatorname{aff}(\operatorname{supp}G)$ has dimension $m-1$. By (a), $\operatorname{supp}(G)\subseteq\mathcal M^*\subseteq\operatorname{aff}(h^*_1,\dots,h^*_m)$, whose dimension is at most $m-1$; hence $\operatorname{aff}(h^*_1,\dots,h^*_m)=\operatorname{aff}(\operatorname{supp}G)$ and the $h^*_j$ are affinely independent. The same argument applied to any $(\bar h_1,\dots,\bar h_m)\in\mathcal H^m_G$ gives $\operatorname{aff}(\bar h_1,\dots,\bar h_m)=\operatorname{aff}(\operatorname{supp}G)$; since the $\bar h_j$ are linearly independent, $[0]_{B^2}\notin\operatorname{aff}(\operatorname{supp}G)=\operatorname{aff}(h^*_1,\dots,h^*_m)$. Affine independence together with $[0]_{B^2}\notin\operatorname{aff}(h^*_1,\dots,h^*_m)$ yields linear independence. Hence $\{h_j^*\}_{j=1}^m\in \mathcal{H}_G^m$.
\qed 
 
\subsection{A map from vertices to proportions}
We next aim to identify the probability distribution $D$ of the random proportion $\boldsymbol{p}\sim D$, knowing only $G$ and the set of vertices $\{h_j\}_{j=1}^m$. Consequently, we are interested in measuring the distance between proportions in terms of their distributions.
 
For this reason, we equip $\mathfrak{D}(S^m)$ with the Wasserstein-1 metric $d_D$ defined by
\[
d_D(\boldsymbol{p},\boldsymbol{p}')
=
\inf_{\gamma\in\Gamma(D,D')}
\int_{S^m\times S^m}\|\boldsymbol{q}-\boldsymbol{q}'\|\,d\gamma(\boldsymbol{q},\boldsymbol{q}'),
\]
where $\|\cdot\|$ denotes the Euclidean norm on $\mathbb{R}^m \supset S^m$, and $\Gamma(D,D')$ is the set of joint distributions on $S^m\times S^m$ with marginals $D$ and $D'$. We note a slight abuse of notation in writing $d_D(\boldsymbol{p},\boldsymbol{p}')$, as this formally represents the distance between the distributions $D \sim \boldsymbol{p}$ and $D' \sim \boldsymbol{p}'$, although it is expressed directly in terms of the parameter vectors $\boldsymbol{p}$ and $\boldsymbol{p}'$. This metric quantifies the discrepancy between two random vectors as the minimal expected displacement needed to turn one distribution into the other; in particular, $d_D(\boldsymbol{p},\boldsymbol{p}')\le\mathbb{E}\|\boldsymbol{p}-\boldsymbol{p}'\|$ whenever $\boldsymbol{p}$ and $\boldsymbol{p}'$ are defined on the same probability space. Since $S^m$ is compact, $d_D$ metrizes weak convergence on $\mathfrak{D}(S^m)$ \cite{villani2009, billingsley1995}.
 
Thanks to Proposition \ref{prop:unique_p}, we can define the function $\phi_G: \mathcal{H}_G^m\to \mathfrak{D}(S^m)$, a map that given the distribution of a random mixture $G$ and a proper $m$-tuple of vertices, returns the distribution $D$ of the relative random proportion $\boldsymbol{p}$, namely 
 
\[\phi_G(\{h_j\}_{j=1}^m)=D \quad \textrm{iff}\quad G=^d\bigoplus_{j=1}^m p_j\odot h_j \,\textrm{ for } \,\boldsymbol{p}\sim D.\] 
We can prove that $\phi_G$ is continuous.
 
\begin{proposition}[$\phi_G$ is continuous]
    Let $G$ be a random mixture of densities and let $\phi_G$ be defined as above. Then, $\phi_G$ is continuous.
    \label{prop:continuous_phi_G}
\end{proposition}
\textbf{Proof of Proposition \ref{prop:continuous_phi_G}}: Let $\{h_j\}_{j=1}^m\in\mathcal{H}^m_G$. We want to prove that $(\|\epsilon\|_{\mathcal{H}^m}\to 0) \Rightarrow (d_D(\phi_G(\{h_j\oplus\epsilon_j\}_j),\phi_G(\{h_j\}_j))\to 0)$, for $\epsilon=(\epsilon_1,\dots,\epsilon_m)\in \mathcal{H}^m$ such that $\{h_j\oplus\epsilon_j\}_{j=1}^m\in \mathcal{H}^m_G$. Notice that every $m$-tuple $\{h'_j\}_j\in\mathcal{H}^m_G$ can be written in this form, with $\epsilon_j=h'_j\ominus h_j\in\mathcal{H}$, since $\mathcal{H}$ is a subspace of $B^2(I)$.
 
Since $\mathbb{P}(G\in\mathcal{M}(\{h_j\}_j))=\mathbb{P}(G\in\mathcal{M}(\{h_j\oplus\epsilon_j\}_j))=1$, for $\mathbb{P}$-almost every $\omega\in\Omega$ there exist $\boldsymbol{p}(\omega),\boldsymbol{p}^{(\epsilon)}(\omega)\in S^m$ such that
\[
G(\omega)=\bigoplus_{j=1}^m p_j(\omega)\odot h_j=\bigoplus_{j=1}^m p^{(\epsilon)}_j(\omega)\odot (h_j\oplus\epsilon_j)=\]
\[=\bigoplus_{j=1}^m p^{(\epsilon)}_j(\omega)\odot h_j\oplus\bigoplus_{j=1}^m p^{(\epsilon)}_j(\omega)\odot\epsilon_j,
\]
meaning that both the original and the deviated pair of vertices and proportions still represent $G$. By linear independence of $\{h_j\}_j$ and of $\{h_j\oplus\epsilon_j\}_j$ (Definition \ref{def:H^m_g}), these coefficients are unique and depend measurably on $G(\omega)$; hence $\boldsymbol{p}$ and $\boldsymbol{p}^{(\epsilon)}$ are random vectors defined on the same probability space as $G$, and by Proposition \ref{prop:unique_p}, $\boldsymbol{p}\sim\phi_G(\{h_j\}_j)$ and $\boldsymbol{p}^{(\epsilon)}\sim\phi_G(\{h_j\oplus\epsilon_j\}_j)$.
 
The above equality is true if and only if:
\[
\bigoplus_{j=1}^m p_j(\omega)\odot h_j\ominus\bigoplus_{j=1}^m p^{(\epsilon)}_j(\omega)\odot h_j=\bigoplus_{j=1}^m p^{(\epsilon)}_j(\omega)\odot\epsilon_j
\]
If and only if:
\[
\bigoplus_{j=1}^m \big(p_j(\omega)-p^{(\epsilon)}_j(\omega)\big)\odot h_j=\bigoplus_{j=1}^m p^{(\epsilon)}_j(\omega)\odot\epsilon_j
\]
Since $0\le p^{(\epsilon)}_j\le 1$, by the triangle inequality and the Cauchy--Schwarz inequality the norm of the right-hand side is bounded by $\sum_{j=1}^m\|\epsilon_j\|_{B^2(I)}\le\sqrt{m}\,\|\epsilon\|_{\mathcal{H}^m}$. On the other hand, since $h_1,\dots,h_m$ are linearly independent, the map $\boldsymbol{x}\mapsto\|\bigoplus_{j=1}^m x_j\odot h_j\|_{B^2(I)}$ is continuous and strictly positive on the unit sphere of $\mathbb{R}^m$, which is compact; hence, by homogeneity, there exists $c>0$, depending only on $\{h_j\}_j$, such that $\|\bigoplus_{j=1}^m x_j\odot h_j\|_{B^2(I)}\ge c\,\|\boldsymbol{x}\|$ for every $\boldsymbol{x}\in\mathbb{R}^m$. Applying this inequality to $\boldsymbol{x}=\boldsymbol{p}(\omega)-\boldsymbol{p}^{(\epsilon)}(\omega)$, we obtain, for $\mathbb{P}$-almost every $\omega\in\Omega$,
\[
\|\boldsymbol{p}(\omega)-\boldsymbol{p}^{(\epsilon)}(\omega)\|\le \frac{\sqrt{m}}{c}\,\|\epsilon\|_{\mathcal{H}^m}.
\]
Finally, denote by $\gamma^*$ the joint distribution of $(\boldsymbol{p},\boldsymbol{p}^{(\epsilon)})$. Since $\boldsymbol{p}\sim\phi_G(\{h_j\}_j)$ and $\boldsymbol{p}^{(\epsilon)}\sim\phi_G(\{h_j\oplus\epsilon_j\}_j)$, the marginals of $\gamma^*$ are exactly these two distributions, i.e.\ $\gamma^*\in\Gamma\big(\phi_G(\{h_j\}_j),\phi_G(\{h_j\oplus\epsilon_j\}_j)\big)$. As $d_D$ is defined as an infimum over all such joint distributions, it is bounded by the value attained at $\gamma^*$:
\[
d_D(\boldsymbol{p},\boldsymbol{p}^{(\epsilon)})
\le\int_{S^m\times S^m}\|\boldsymbol{q}-\boldsymbol{q}'\|\,d\gamma^*(\boldsymbol{q},\boldsymbol{q}')
=\mathbb{E}\|\boldsymbol{p}-\boldsymbol{p}^{(\epsilon)}\|.
\]
Taking the expectation of the almost-sure bound obtained above, we conclude
\[
d_D(\boldsymbol{p},\boldsymbol{p}^{(\epsilon)})\le\frac{\sqrt{m}}{c}\,\|\epsilon\|_{\mathcal{H}^m},
\]
which tends to $0$ as $\|\epsilon\|_{\mathcal{H}^m}\to 0$.\\
\qed
\\
\medskip 
Note that $\phi_G$ is not injective in general: any enlargement of a valid convex hull yields another valid representation of $G$, with a different
distribution of the proportions (Figure~\ref{intro_mix}).
\vspace{1cm}

Let $\boldsymbol x=\psi(\boldsymbol p)\in\mathbb R^{m-1}$ denote the coordinates of $\boldsymbol p$ on the $(m-1)$-dimensional simplex, for a given function $\psi$, and let $\Sigma_p=\operatorname{Cov}(\boldsymbol x)$, seen as a function of the distribution of $\boldsymbol p$. We assume $\Sigma_p$ to be finite and continuous with respect to $d_D$. This holds trivially for bounded coordinates, and for the ilr transformation (see Section~\ref{pmlu}) whenever the proportions do not accumulate increasingly heavy mass near the boundary of the simplex.

As discussed in Section~\ref{identifiability}, maximising $\operatorname{tr}(\Sigma_p)$ selects the representation whose proportions spread out most within the convex hull. While related in concept, maximising the trace encourages the distribution of $\boldsymbol{p}$ to spread out evenly within the convex hull, thereby utilising the available space efficiently. By contrast, directly minimising the volume of the convex hull tends to concentrate the vertices near the most extreme observations, which may distort the recovered statistical structure by overweighting atypical data points.

Figure~\ref{figs:min_element} schematically illustrates the concept of a minimal element: it corresponds to the statistically space-efficient convex hull, i.e.\ to the representation in which the proportions exhibit the largest dispersion. We aim to prove that a minimal element with respect to $<_G$ exists and is unique.
 
\subsection{Ordering representations}
Recall from Section~\ref{identifiability} the equivalence relation $\sim_\pi$, which identifies $m$-tuples of vertices differing only by a permutation $\pi$ of $\{1,\dots,m\}$, the set $\mathcal C_G$ of valid representations, and the order $<_G$, defined by $(\{h_j\}_j,\boldsymbol p)<_G(\{\bar h_j\}_j,\bar{\boldsymbol p})$ whenever $\operatorname{tr}(\Sigma_p)>\operatorname{tr}(\Sigma_{\bar p})$. We first recall the order-theoretic notions needed to study it.

\begin{definition}[Strict partial order, \cite{Simovici2008}]
    A relation $<$ on a set $X$ is a strict partial order  if it is:
    \begin{itemize}
        \item Irreflexive: $\forall x\in X$, $x\not <x$.
        \item Asymmetric: $\forall x,y\in X$ if $x<y$ then $y\not <x$.
        \item Transitive: $\forall x,y,z\in X$ if $(x<y) \wedge(y<z)$ then $x<z$.
    \end{itemize}
\end{definition}

\begin{definition}[Minimality \cite{Simovici2008}]
    An element $x\in (X, <)$, where $<$ is a (generic) order relation is called minimal if $\not\exists y\in X: y<x$.
\end{definition}

    \begin{figure}[!htbp]
        \centering
                \includegraphics[scale=0.5]{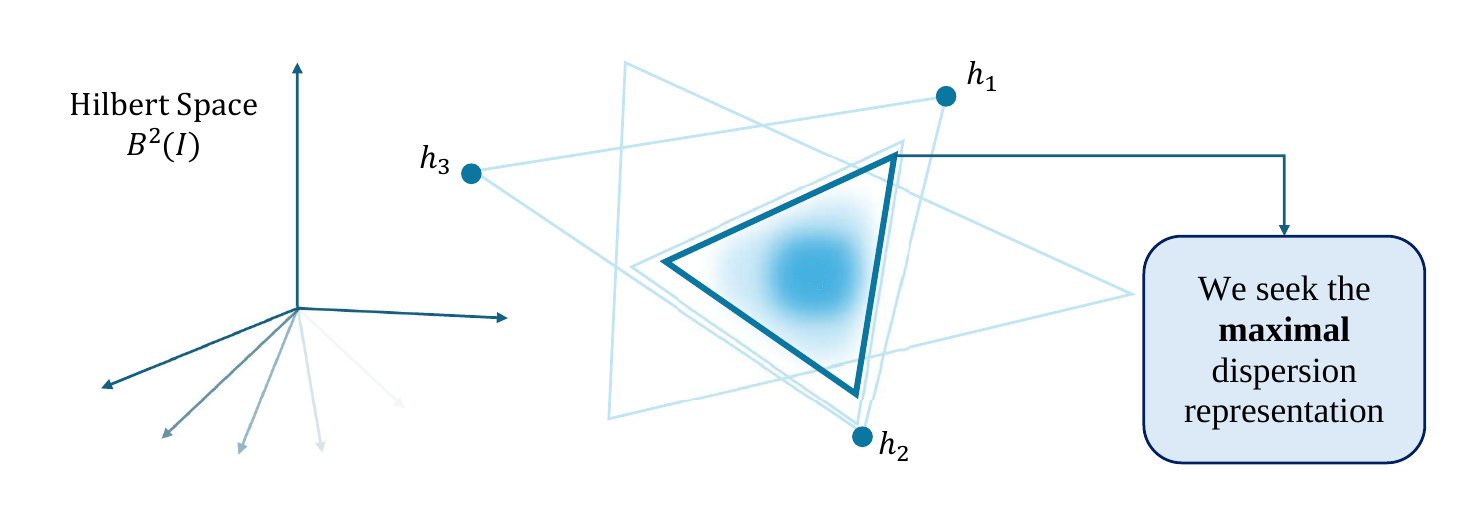}
                \caption{Schematic illustration of the concept of a minimal element: it corresponds to the convex hull in which the associated proportions exhibit the largest dispersion.
}
                \label{figs:min_element}
    \end{figure}

Let $\sim_\pi$ denote the equivalence relation on $m$-tuples of vertices defined as follows: two $m$-tuples $(h_1,\dots,h_m)$ and $(h_1',\dots,h_m')$ are equivalent, $(h_1,\dots,h_m) \sim_\pi\, (h_1',\dots,h_m')$, if there exists a permutation $\pi$ of $\{1,\dots,m\}$ such that
\[
(h_1,\dots,h_m) = (h'_{\pi(1)},\dots,h'_{\pi(m)}).
\]
 Let $\mathcal{H}_G^m\times \mathfrak{D}(S^m)/\sim_\pi$ be the quotient space with respect to $\sim_\pi$.

\begin{proposition}
\label{prop:strict_order}
    The relation $<_{G}$ is a \textbf{strict partial order} relation on $\mathcal{C}_{G}\subset \mathcal{H}_G^m\times \mathfrak{D}(S^m)/\sim_\pi$, where $\mathfrak{D}(S^m)$ is the space of all possible distributions on $S^m$, $\mathcal{H}_G^m\times \mathfrak{D}(S^m)/\sim_\pi$ is the quotient space w.r.t. the equivalence relation $\pi$ and \[\mathcal{C}_G=\bigg\{(\{h_j\}_{j=1}^m,\boldsymbol{p})\in \mathcal{H}_G^m\times \mathfrak{D}(S^m): G\operatorname*{=}^d \bigoplus_{j=1}^m p_j\odot h_j\bigg\}\bigg/\sim_\pi =\]
    \[
    =\bigg\{(\{h_j\}_{j=1}^m,\boldsymbol{p})\in \mathcal{H}_G^m\times \mathfrak{D}(S^m) : \boldsymbol{p}=\phi_G(\{h_j\}_{j=1}^m)\bigg\}\bigg/\sim_\pi\ .
    \]
    
\end{proposition}
\textbf{Proof of Proposition \ref{prop:strict_order}}: We only have to check all the properties of strict partial orders:
 \begin{enumerate}
    \item Irreflexive: trivial
    \item Asymmetric: $(\{h_j\}_{j=1}^m,\boldsymbol{p})<_{G}(\{\bar h_j\}_{j=1}^{m},\boldsymbol{\bar p})$ implies that $\operatorname{tr}(\Sigma_p)> \operatorname{tr}(\Sigma_{\bar p})$, hence it cannot be true that $(\{h_j\}_{j=1}^m,\boldsymbol{p})>_{G}(\{\bar h_j\}_{j=1}^{m},\boldsymbol{\bar p})$.

    \item Transitive: let $(\{h_j\}_{j=1}^m,\boldsymbol{p})<_{G}(\{\bar h_j\}_{j=1}^{m},\boldsymbol{\bar p})$ and $(\{\bar h_j\}_{j=1}^{m},\boldsymbol{\bar p})<_{G}(\{\hat h_j\}_{j=1}^{ m},\boldsymbol{\hat p})$. Then $\operatorname{tr}(\Sigma_p)>\operatorname{tr}(\Sigma_{\bar p})>\operatorname{tr}(\Sigma_{\hat p})$. This implies that $(\{h_j\}_{j=1}^m,\boldsymbol{p})<_{G}(\{\hat h_j\}_{j=1}^{ m},\boldsymbol{\hat p})$.

\end{enumerate}
\qed \\

\paragraph{Existence of a minimal element}
By Definition~\ref{def:tight}, containment (T1) confines the search for a minimal element to $\mathcal M_0$, a compact set, while monotonicity (T2) relates statistical and geometric minimality (see Figure~\ref{fig:tightness}).

\begin{proposition}[Existence of a minimal element]
\label{prop:min_exist}
If $G$ is non-degenerate and satisfies (T1), then $(\mathcal C_G,<_G)$ has at least one minimal element.
\end{proposition}

\textbf{Proof of Proposition \ref{prop:min_exist}}: Consider a sequence in $\mathcal C_G$ whose traces converge to $M$. If $M=\operatorname{tr}(\Sigma_{p_0})$, the representation associated with $\mathcal M_0$ is already minimal. Otherwise, by (T1), the convex hulls of the sequence are
eventually contained in $\mathcal M_0$. Since $\mathcal M_0$ is compact and $\mathcal H^m$ is finite-dimensional, the vertices converge, up to a subsequence, to some $\{h^*_j\}_j$, which lies in $\mathcal H^m_G$ by Proposition~\ref{prop:close_H_G^m}. By Proposition~\ref{prop:continuous_phi_G}, the corresponding distributions of the proportions converge weakly to $\phi_G(\{h^*_j\}_j)$; by continuity of the trace, the limit representation has trace $M$. Since no
element of $\mathcal C_G$ has a larger trace, it is minimal.
\qed

So far, condition (T1) guarantees the existence of a minimal representation of $G$ with respect to $<_G$. Its uniqueness, and its identification with the geometrically minimal representation, rely on (T2) and the Probabilistic Pure Pixel (3P) assumption (Definition~\ref{def:3P}). Intuitively, 3P requires each vertex to be approachable with positive probability, up to arbitrary precision: no vertex is probabilistically hidden. Under these conditions, Theorem~\ref{minimal_theorem} (Section~\ref{identifiability}) states that the minimal representation is unique.
 
\begin{figure}[H]
    \centering
    \includegraphics[width=0.4\linewidth]{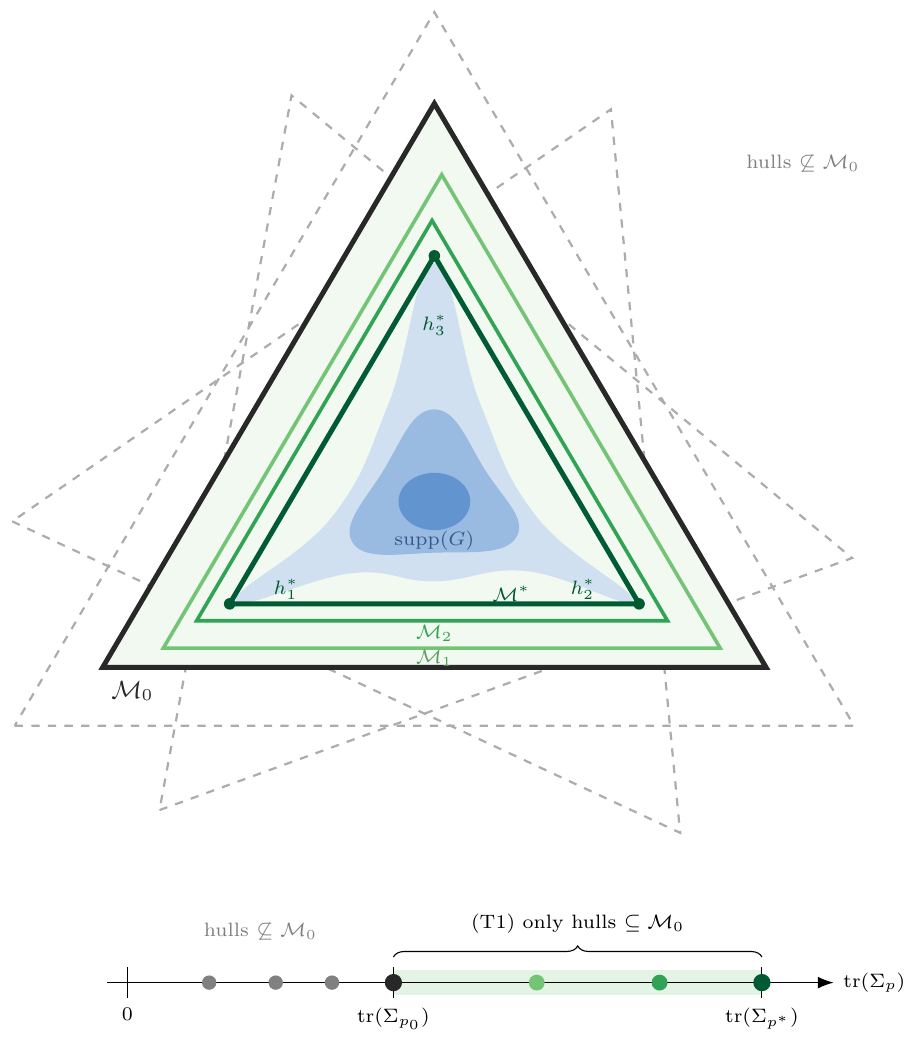}
\caption{Schematic illustration of tightness (Definition~\ref{def:tight}).
\textbf{(T1) Containment:} every representation with $\operatorname{tr}(\Sigma_p)\ge\operatorname{tr}(\Sigma_{p_0})$ has its convex hull inside $\mathcal M_0$ (black); hulls not contained in $\mathcal M_0$ (grey, dashed) are less dispersed.
\textbf{(T2) Monotonicity:} inside $\mathcal M_0$, shrinking a hull increases the trace, e.g.\ $\mathcal M^*\subsetneq\mathcal M_2\subsetneq\mathcal M_1\subsetneq\mathcal M_0$ implies $\operatorname{tr}(\Sigma_{p^*})>\operatorname{tr}(\Sigma_{p_2})>\operatorname{tr}(\Sigma_{p_1})>\operatorname{tr}(\Sigma_{p_0})$.
$\mathcal M^*$ denotes the minimal hull.
The bottom axis orders the representations by $\operatorname{tr}(\Sigma_p)$. See Remark~\ref{rem:tightness} for an intuition on when tightness holds.}
\label{fig:tightness}
\end{figure}
\textbf{Proof of Theorem \ref{minimal_theorem}}: By Proposition~\ref{prop:min_exist}, a minimal element exists. Let $\mathcal M'$ be the convex hull of any minimal element, $\mathcal M^*=\mathcal M(h^*_1,\dots,h^*_m)$ with $\{h^*_j\}_j$ as in Definition~\ref{def:3P}, and let $\Theta=\mathcal M^*\cap\mathcal M'$, which is closed and contains $\operatorname{supp}(G)$. If some $h^*_j\in\Theta^c$, which is open, there exists $\epsilon>0$
with $D_\epsilon(h^*_j)\subset\Theta^c$, hence $\mathbb P\big(G\in D_\epsilon(h^*_j)\big)=0$, contradicting 3P. Therefore $h^*_j\in\Theta$ for every $j$ and, by convexity, $\mathcal M^*=\Theta\subseteq\mathcal M'$. Since $\mathcal M'$ is minimal, its trace is $M\ge\operatorname{tr}(\Sigma_{p_0})$, so $\mathcal M'\subseteq\mathcal M_0$ by (T1). If $\mathcal M^*\subsetneq\mathcal M'$, (T2) would give a trace larger than $M$, which is impossible. Hence $\mathcal M^*=\mathcal M'$: the vertices of a simplex are its extreme points, so the two $m$-tuples coincide up to a permutation, and so do the distributions of the proportions by Proposition~\ref{prop:unique_p}. Thus every minimal element coincides with $(\{h^*_j\}_j,\boldsymbol p^*)$.
\qed

\section{Proofs of Section \ref{pmlu}}
\label{appendix_proofs5}
\textbf{Proof of Equation \ref{optH}}: First, we can rewrite Equation \ref{hatl}, fixing $\Sigma_\epsilon, \boldsymbol \mu_p,\Sigma_p$ as follow:
\begin{equation}
    \begin{split}
       \hat l(\mathbb H)=&-\frac{1}{2}\sum_{b=1}^B\sum_{i=1}^n w_{b,i}\cdot( \boldsymbol c_i-\mathbb H\cdot \boldsymbol p_{b,i})'\Sigma_\epsilon^{-1}( \boldsymbol c_i-\mathbb H\cdot \boldsymbol p_{b,i})+\\
       &- \lambda_{sm}\operatorname{tr}(\mathbb H' \cdot D_2\cdot\mathbb H)+const 
    \end{split}
\end{equation}

where $\boldsymbol p_{b,i}$ is the $b$-th sampled proportion for the $i-th$ unit, $\boldsymbol p_{b,i}$ sampled according to the fixed $\boldsymbol \mu_p$, and $\Sigma_p$. Notice that $\forall b,i: \dim(\boldsymbol p_{b,i})=m\times 1$. Lastly, $D_2=\bigg [\langle\zeta_r'',\zeta_s''\rangle\bigg]_{r,s}$, where $\zeta_r$ is the $r$-th element of the $k$-dimensional basis of $\mathcal{H}$. \\

The gradient of $\hat l$ with respect to the matrix $\mathbb H$ has the same dimensions as $\mathbb H$, being the codomain of $\hat l$ one-dimensional, and it is equal to \[\nabla_{\mathbb H} \hat l=\sum_{b=1}^B\sum_{i=1}^n w_{b,i}\cdot\Sigma_\epsilon^{-1}( \boldsymbol c_i-\mathbb H\cdot \boldsymbol p_{b,i})\boldsymbol p_{b,i}'-2\lambda_{sm}(
D_2\cdot \mathbb H)\]
\[=\Sigma_\epsilon^{-1}\bigg(\sum_{b=1}^B\sum_{i=1}^n w_{b,i}( \boldsymbol c_i-\mathbb H\cdot \boldsymbol p_{b,i})\boldsymbol p_{b,i}'\bigg)-2\lambda_{sm}(D_2 \cdot \mathbb H)\]
By imposing the gradient to be equal to the $k\times m $ dimensional null matrix, we get the following expression:
\[
\Sigma_\epsilon^{-1}\bigg(\sum_{b=1}^B\sum_{i=1}^n w_{b,i}\cdot \boldsymbol c_i\boldsymbol p_{b,i}'\bigg)=\]
\[=\Sigma_\epsilon^{-1}\mathbb H\bigg(\sum_{b=1}^B\sum_{i=1}^n w_{b,i}\cdot(\boldsymbol p_{b,i}\boldsymbol p_{b,i}')\bigg)+2\lambda_{sm}(D_2 \cdot \mathbb H)
\]
That implies
\[
\bigg(\sum_{b=1}^B\sum_{i=1}^n w_{b,i}\cdot \boldsymbol c_i\boldsymbol p_{b,i}'\bigg)=\]
\[=\mathbb H\bigg(\sum_{b=1}^B\sum_{i=1}^n w_{b,i}\cdot(\boldsymbol p_{b,i}\boldsymbol p_{b,i}')\bigg)+(2\lambda_{sm} \cdot \Sigma_\epsilon \cdot D_2) \cdot \mathbb H
\]
Hence, the problem is equivalent to finding $\mathbb H$ such that:
\[ A_l \mathbb H +\mathbb HA_r=C\]
and the solution is:

\[
\operatorname{vec}(\mathbb H^*)=(A_r'\otimes I_{k\times k}+I_{m\times m}\otimes A_l)^{-1}\operatorname{vec}(C)
\]
where $\otimes$ denotes the Kronecker product, $I_{k\times k}$ and $I_{m\times m}$ are, respectively, the identity matrix of dimension $k\times k$ and $m\times m$, and $\operatorname{vec}(\cdot)$ is the operation that vectorizes the matrices by column.\\
\qed

Notice that if $\lambda_{sm}=0$, the maximum of $\hat{l}$ -- see Equation \ref{hatl} -- with respect to $\mathbb H$ is equal to:

\[
\mathbb H=\bigg(\sum_{b=1}^B\sum_{i=1}^n w_{b,i}\cdot \boldsymbol c_i\boldsymbol p_{b,i}'\bigg)\bigg(\sum_{b=1}^B\sum_{i=1}^n w_{b,i}\cdot\boldsymbol p_{b,i}\boldsymbol p_{b,i}'\bigg)^{-1}=C A_r^{-1}
\]

The matrix $C$ can be interpreted as a weighted mean of the observations $\{f_i\}$, where the weights are given by their posterior proportions. Let us examine the inverse of the second term, namely $A_r^{-1}$.

\begin{itemize}
    \item The diagonal entries of $A_r$ represent average squared posterior proportions associated with each vertex.
    \item The off-diagonal entries, say in position $(j,j')$, quantify the joint contribution of vertices $j$ and $j'$, and can be interpreted as a \textit{probability} of their coexistence within the data. Such coexistence indices tend to increase when many units $i$ exhibit a uniform coexistence in two or more vertices. This is the case when the number of vertices $m$ is overestimated.
\end{itemize}

Consequently, when $A_r^{-1}$ is applied to $C$, it acts as a correction term that mitigates the artificial coexistence between different vertices, thereby improving their identifiability. This mechanism is analogous to the decorrelation of a dataset, which is obtained by applying the inverse of the upper triangular Cholesky factor of the covariance matrix.
\vspace{0.6cm}\\
\noindent
\textbf{Proof of Equation \ref{opt_sigeps}}: 
First, we can rewrite Equation \ref{hatl}, fixing $\mathbb H, \boldsymbol \mu_p,\Sigma_p$ as follow:
\begin{equation}
    \begin{split}
      \hat l(\Sigma_\epsilon)=&-\frac{1}{2}\sum_{b=1}^B\sum_{i=1}^n w_{b,i}\cdot( \boldsymbol c_i-\mathbb H\cdot \boldsymbol p_{b,i})'\Sigma_\epsilon^{-1}( \boldsymbol c_i-\mathbb H\cdot \boldsymbol p_{b,i})+\\
      &-\log(\operatorname{det}(\Sigma_\epsilon))+const  
    \end{split}
\end{equation}

this maximization problem is equivalent to finding the maximum likelihood covariance matrix when the mean is fixed \cite{bishop_book}, whose solution is:

\begin{equation}
\small
    \begin{split}
        &\Sigma_\epsilon^*=\\
        &=\frac{1}{\sum_b\sum_iw_{b,i}}\sum_{b=1}^B\sum_{i=1}^nw_{b,i}((\boldsymbol c_i-\mathbb H \boldsymbol p_{b,i})- \boldsymbol0)((\boldsymbol c_i-\mathbb H \boldsymbol p_{b,i})- \boldsymbol0)'=\\
&=\frac{1}{\sum_b\sum_iw_{b,i}}\sum_{b=1}^B\sum_{i=1}^nw_{b,i}(\boldsymbol c_i-\mathbb H \boldsymbol p_{b,i})(\boldsymbol c_i-\mathbb H \boldsymbol p_{b,i})'
    \end{split}
\end{equation}

\qed
\vspace{0.6cm}\\
\noindent
\textbf{Proof of Equation \ref{opt_mup}}: this result is well known and employed in modern statistics, e.g. in \cite{steyer2023}.
\qed 
\vspace{0.6cm}\\
\noindent
\textbf{Proof of Equation \ref{opt_sigp}}: 
\[
\nabla_{{\Sigma_p}^{-1}}\ \hat l=\]
\[ = \nabla_{{\Sigma_p}^{-1}} \bigg(\sum_b\sum_i w_{b,i} (\psi(p_{b,i})-\boldsymbol\mu_p)'\Sigma_p^{-1}(\psi(p_{b,i})-\boldsymbol\mu_p)+\]
\[-\log(\det({\Sigma_p}))-\lambda_{tr}\operatorname{tr}(\Sigma_p^{-1})\bigg)=
\]
\[
=\nabla_{{\Sigma_p}^{-1}} \bigg(\sum_b\sum_i w_{b,i} (\psi(p_{b,i})-\boldsymbol\mu_p)'\Sigma_p^{-1}(\psi(p_{b,i})-\boldsymbol\mu_p)+\]
\[+(\sum_b\sum_i w_{b,i} )\log(\det({\Sigma_p}^{-1}))-\lambda_{tr}\operatorname{tr}(\Sigma_p^{-1})\bigg)=
\]
\[
= - \sum_b\sum_iw_{b,i}(\boldsymbol \psi(p_{b,i})-\boldsymbol \mu_p)(\boldsymbol \psi(p_{b,i})-\boldsymbol \mu_p)'+\]
\[+(\sum_b\sum_i w_{b,i} )\frac{\det(\Sigma_p^{-1})\Sigma_p}{\det(\Sigma_p^{-1})} - \lambda_{tr}I_{(m-1)\times (m-1)}
\]
Imposing the gradient equal to 0, we proved our statement. \\
\qed\\

\section{Justified choice of $m$ in the simulations}
\label{app:justify_m}
The scree plots of all simulated datasets (see Figure~\ref{m_check}) clearly display that the explained variance is essentially exhausted beyond this point. Therefore, $m = 3$ would be selected based on standard diagnostic criteria, independently of prior knowledge of the data-generating process. 

\begin{figure}[H]
    \centering    \includegraphics[width=1\linewidth]{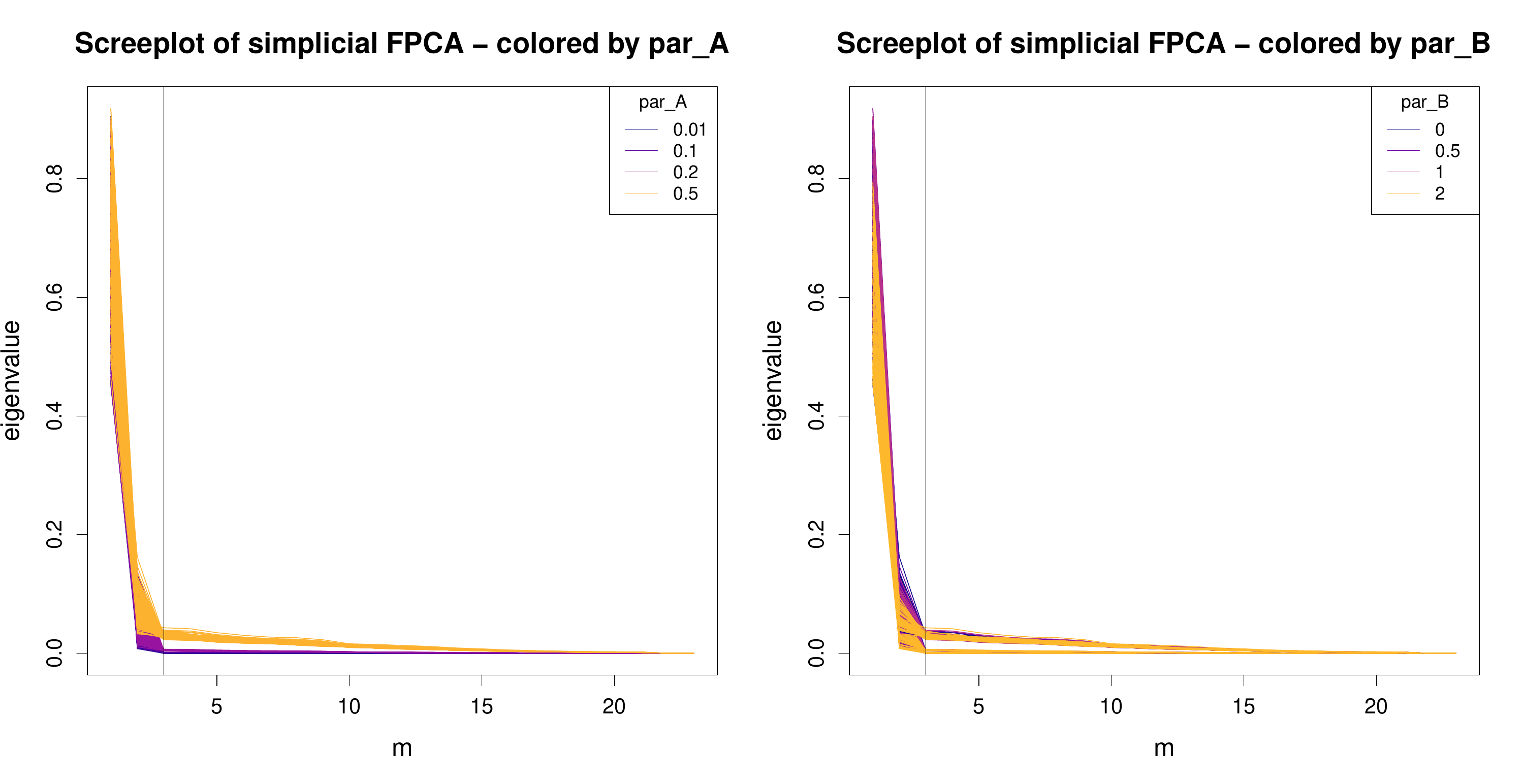}
    \caption{Scree plots of all the simulated datasets -- both Study A and B. As one can notice, the variance is exhausted in $m=3$.}
    \label{m_check}
\end{figure}

\section{Visualisation of the parameter in Study B}
\begin{figure}[H]
    \centering
    \includegraphics[width=0.8\linewidth]{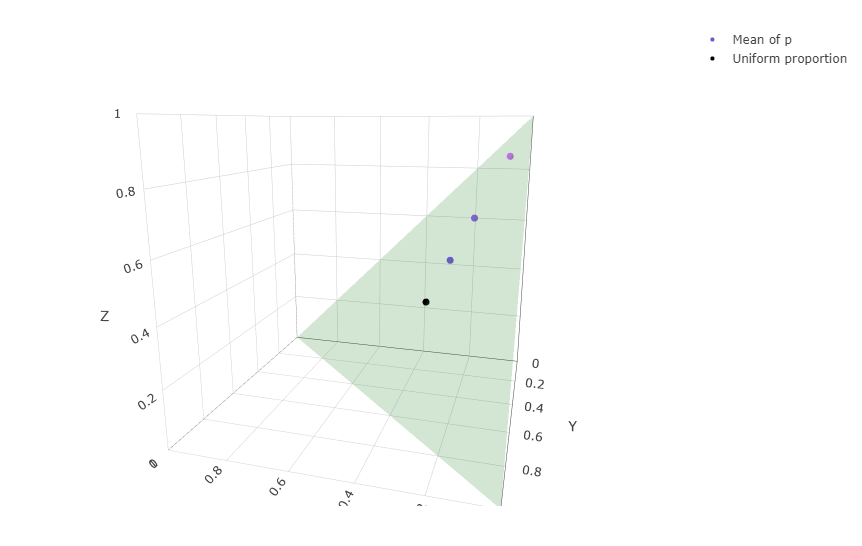}
    \caption{The options of $\mu_p$ in simulation study B, represented in $\mathbb{R}^3$, strictly belonging to $S^3\subset\mathbb{R}^3$. The black point is $\psi^{-1}((0,0))$, the darkest violet is $\psi^{-1}((0,1/2))$ and the lightest is $\psi^{-1}((0,2))$}.
    \label{studyB_mup}
\end{figure}
\section{Effect of the number of vertices on vertex estimation error}
\label{app:n_vertices}
 
In addition to the Study A (noise level, Section~\ref{sec:simA})
and Study B (separation of mean proportions,
Section~\ref{sec:simB}), in this appendix we assess the effect of the
number of vertices $m$ on the estimation accuracy of the BH model. We
compare three configurations: the baseline case $m = 3$ (shared by the
Simulation A and B studies, with $\sigma^2 = 0.01$) and the cases $m = 5$ and $m = 8$ of an additional simulation (Simulation~C), keeping all other parameters fixed
($\sigma^2 = 0.01$, $\mu_p = 0$, $n = 500$, $k = 23$, 50
replicates per configuration). For each replicate we compute the vertex
estimation error as a weighted Frobenius norm between the estimated and
true vertices, after optimal realignment via linear assignment
(as in Section~\ref{sec:simA} and ~\ref{sec:simB}). As discussed below, the results reveal a
trade-off between two competing effects of $m$ on estimation accuracy,
rather than a single monotonic relationship.
 
Figure~\ref{fig:vertex_error_boxplot} shows the distribution of the error
for $m = 3, 5, 8$. Two distinct patterns emerge. First, the spread of the
error across replicates decreases monotonically with $m$: the
interquartile range and whisker extent are markedly larger for $m = 3$
than for $m = 5$ and $m = 8$. Second, the median error level does not
follow the same monotonic trend: it is lowest for $m = 5$, higher for
$m = 3$, and higher still for $m = 8$. The relationship between estimation
error and the number of vertices is therefore non-monotonic, reflecting
two competing effects that we discuss below.
 
\begin{figure}[htbp]
  \centering
  \includegraphics[width=0.65\textwidth]{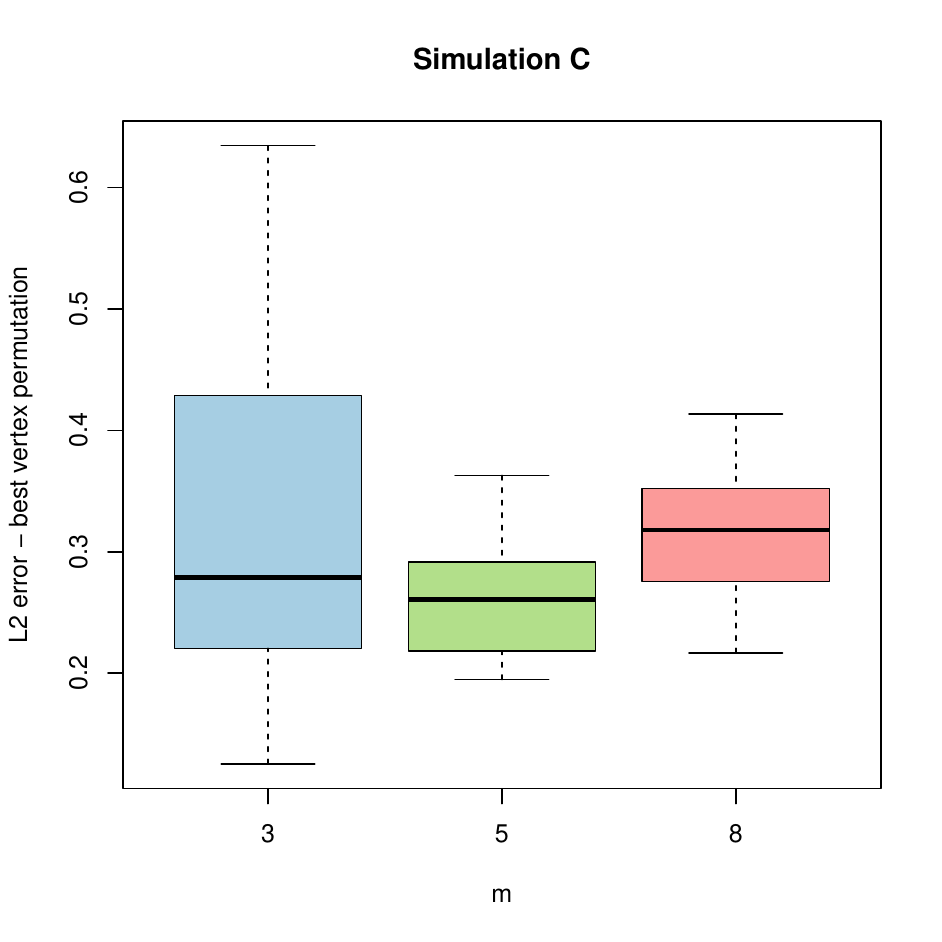}
  \caption{Vertex estimation error (weighted Frobenius norm, outliers not
  shown) for $m = 3, 5, 8$ vertices, BH model, 50 replicates per
  configuration ($\sigma^2 = 0.01$, $\mu_p = 0$, $n = 500$).}
  \label{fig:vertex_error_boxplot}
\end{figure}
 
\paragraph{Effect 1: vertex separation reduces variability} The first
effect is explained by the vertex-generation mechanism in the simulator.
For a given $m$, the $j$-th vertex ($j = 1, \dots, m$) is defined as a
$\mathrm{Beta}(j+1,\, m-j+2)$ density. Figure~\ref{fig:beta_vertices}
shows these densities for $m = 3, 5, 8$. For $m = 3$, the Beta parameters
remain close to the centre ($\mathrm{Beta}(2,4)$, $\mathrm{Beta}(3,3)$,
$\mathrm{Beta}(4,2)$), producing wide, strongly overlapping vertices. As
$m$ increases, the parameters move progressively toward the tails of the
distribution (e.g. $\mathrm{Beta}(2,9)$ and $\mathrm{Beta}(9,2)$ for
$m = 8$), yielding narrower, better-separated vertices. Lower geometric
separation between vertices makes the estimation problem less
identifiable: small perturbations in the simulated data can drive the EM
algorithm toward qualitatively different solutions across replicates. This
effect alone would predict a monotonic decrease in both the level and the
variability of the error as $m$ increases.
 
\begin{figure}[H]
  \centering\includegraphics[width=0.8\textwidth]{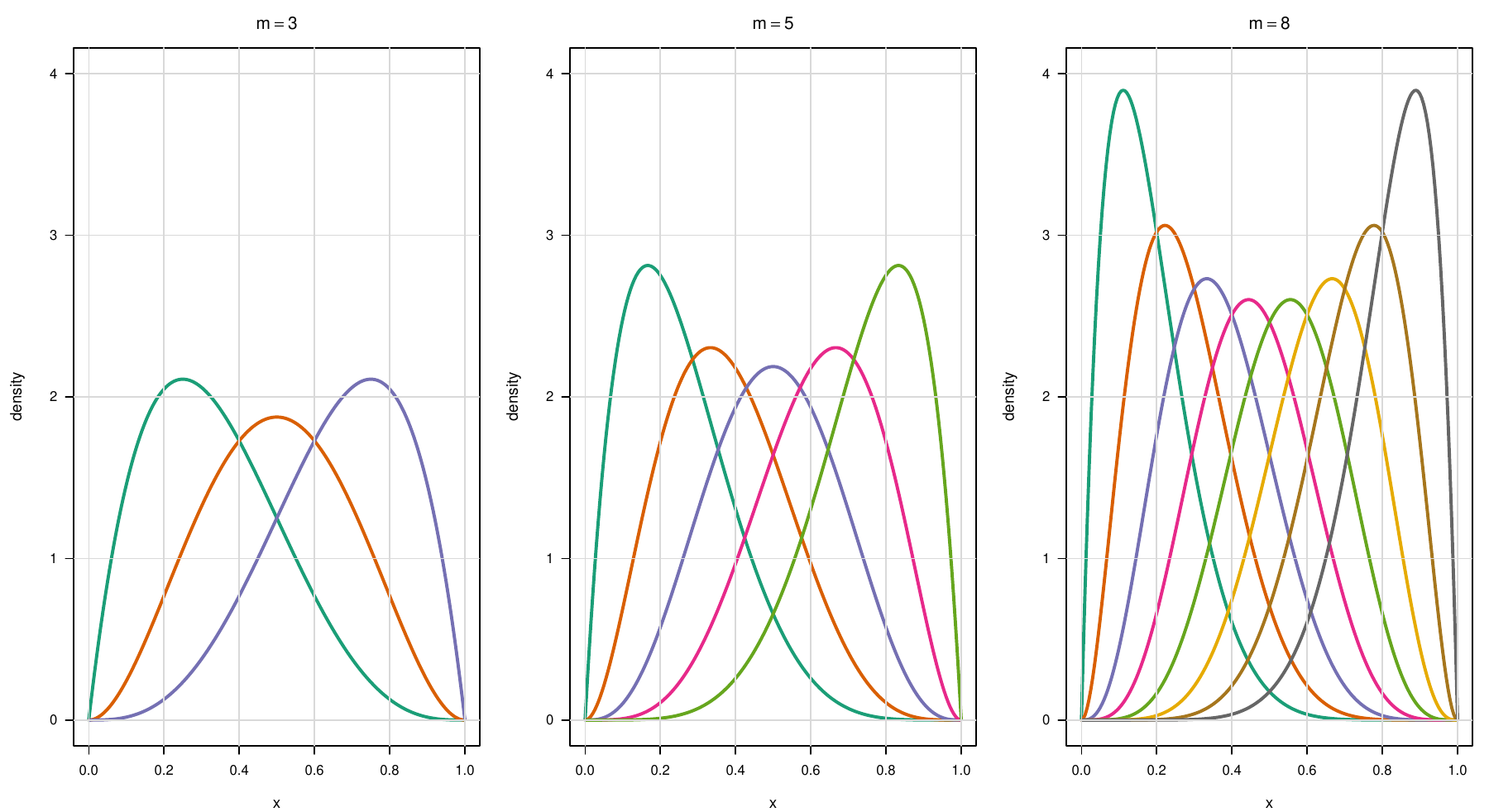}
  \caption{Densities that our vertices approximate, generated by the simulator under the
  $\mathrm{Beta}(j+1,\, m-j+2)$, $j = 1, \dots, m$ parametrisation, for
  $m = 3, 5, 8$. For small $m$ the vertices are wider and overlapping; as
  $m$ increases they become narrower and better separated, reducing the
  ambiguity of the estimation problem and, consequently, the error
  variability shown in Figure~\ref{fig:vertex_error_boxplot}.}
  \label{fig:beta_vertices}
\end{figure}
 
\paragraph{Effect 2: more vertices means more parameters to estimate}
The second effect works in the opposite direction on the median error
level. The vertex matrix $H$ has dimension $k \times m$ with $k = 23$
basis functions, so the number of free parameters to estimate grows
linearly with $m$ (69, 115 and 184 parameters for $m = 3, 5, 8$
respectively), while the sample size is held fixed at $n = 500$ across all
configurations. The ratio of vertex parameters to observations therefore
increases from $0.14$ at $m = 3$ to $0.23$ at $m = 5$ and $0.37$ at
$m = 8$. With a fixed amount of information in the data, estimating more
parameters leaves less information available per parameter, which tends
to inflate the estimation error for each individual vertex. We verified
that this increase is not an artefact of how the error is aggregated
across columns of $H$: computing the error separately for each vertex
(i.e.\ without summing over columns) yields the same pattern, with the
per-vertex error markedly higher at $m = 8$ than at $m = 5$.
 
\paragraph{Net effect} The two mechanisms act on different aspects of the error distribution and, in this experiment, partially offset each other.
Increasing $m$ from 3 to 5 is dominated by the separation effect: vertices become easier to tell apart, and both the level and the variability of the
error decrease. Increasing $m$ further from 5 to 8 is dominated by the parameter-count effect: vertices are even better separated, so the
variability of the error continues to decrease, but the larger number of
parameters relative to the fixed sample size increases the typical
magnitude of the error. The configuration $m = 5$ in this study appears to
sit close to the point at which the two effects balance, combining a
comparatively low error level with low variability across replicates. We
report the results as observed, without correcting for this trade-off, as
it reflects a genuine property of the estimation problem under a fixed
sample size rather than an instability of the proposed method.

\section{Additional figures for the AVIRIS Indian Pines case study}
\label{app:aviris_figures}

This appendix collects auxiliary figures for the AVIRIS Indian Pines case study of Section~\ref{sec:aviris}. Figure~\ref{app:pca_var} reports the diagnostic output of the simplicial FPCA used to select $m=4$, while Figure~\ref{app:vca_casestudy} displays the abundance maps obtained by the VCA competitor, included here for completeness.

\begin{figure}[H]
    \centering
    \begin{minipage}{0.3\textwidth}
        \centering
        \includegraphics[width=\linewidth]{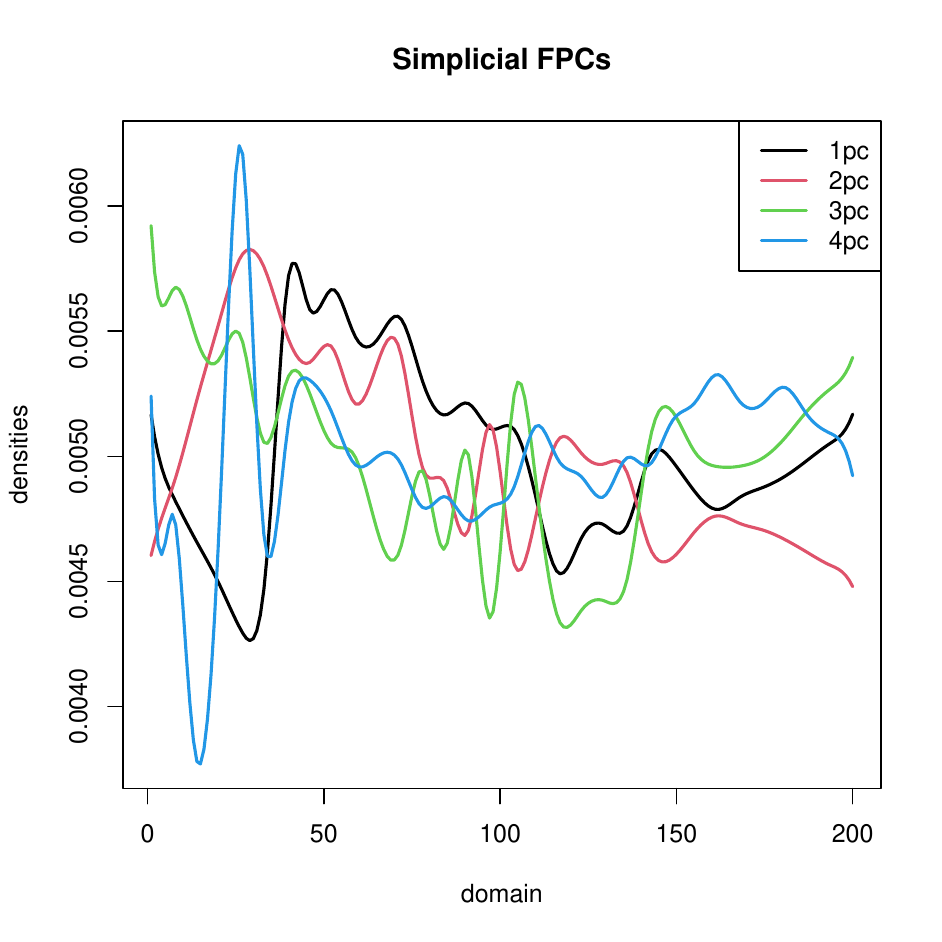}
    \end{minipage}\hfill
    \begin{minipage}{0.3\textwidth}
        \centering
        \includegraphics[width=\linewidth]{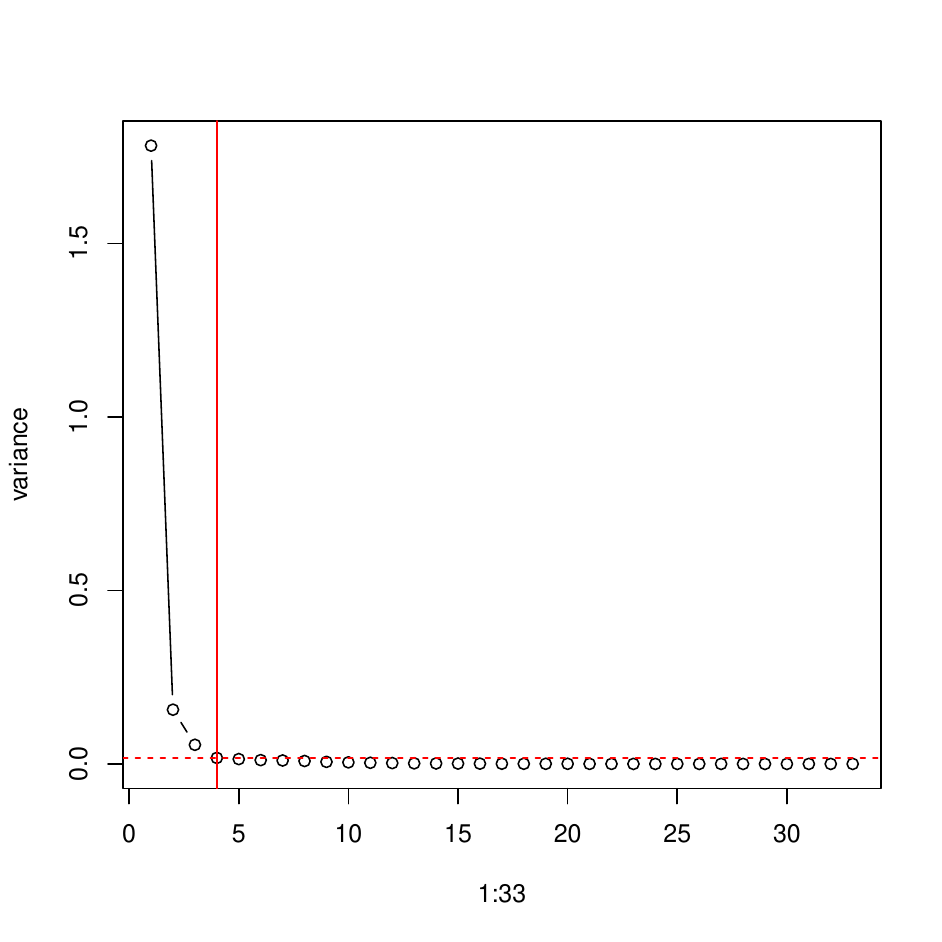}
    \end{minipage}
    \hfill
    \begin{minipage}{0.3\textwidth}
        \centering
        \includegraphics[width=\linewidth]{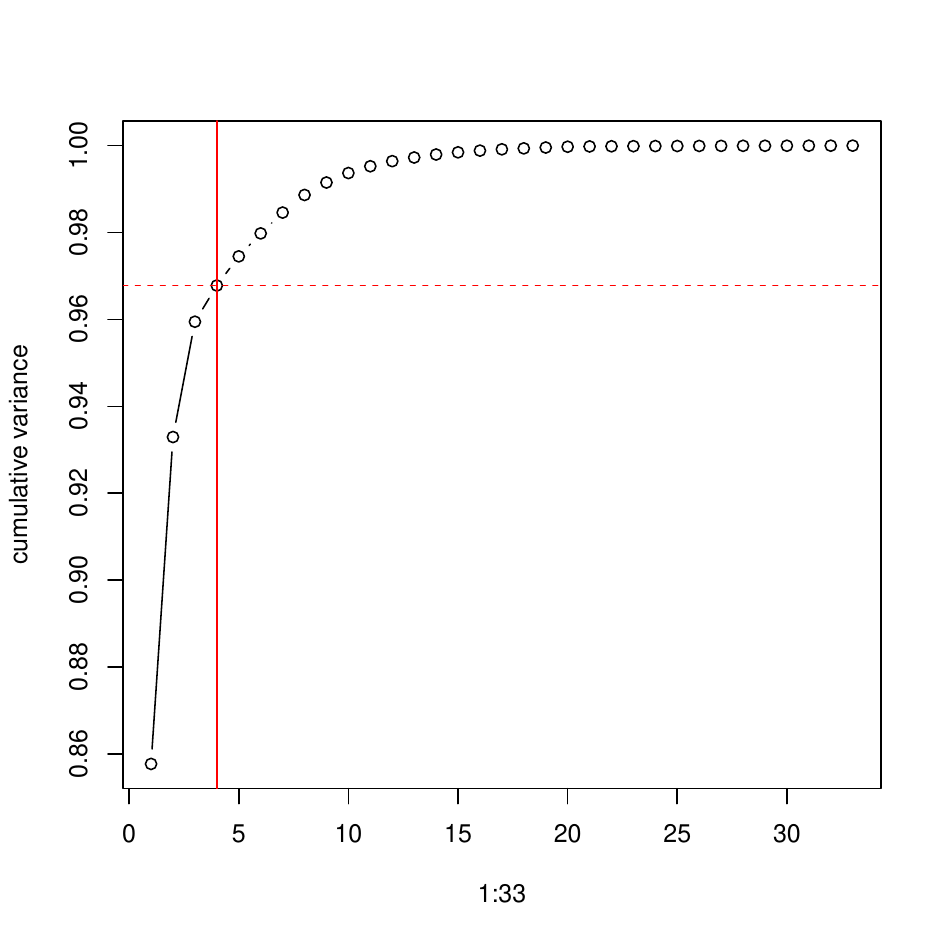}
    \end{minipage}
    \caption{First 4 principal components (left).  Variance explained by the first 4 components (center). Cumulative variance explained by the first 4 components (right).}
    \label{app:pca_var}
\end{figure}

\begin{figure}[H]
    \centering
    \includegraphics[width=0.7\linewidth]{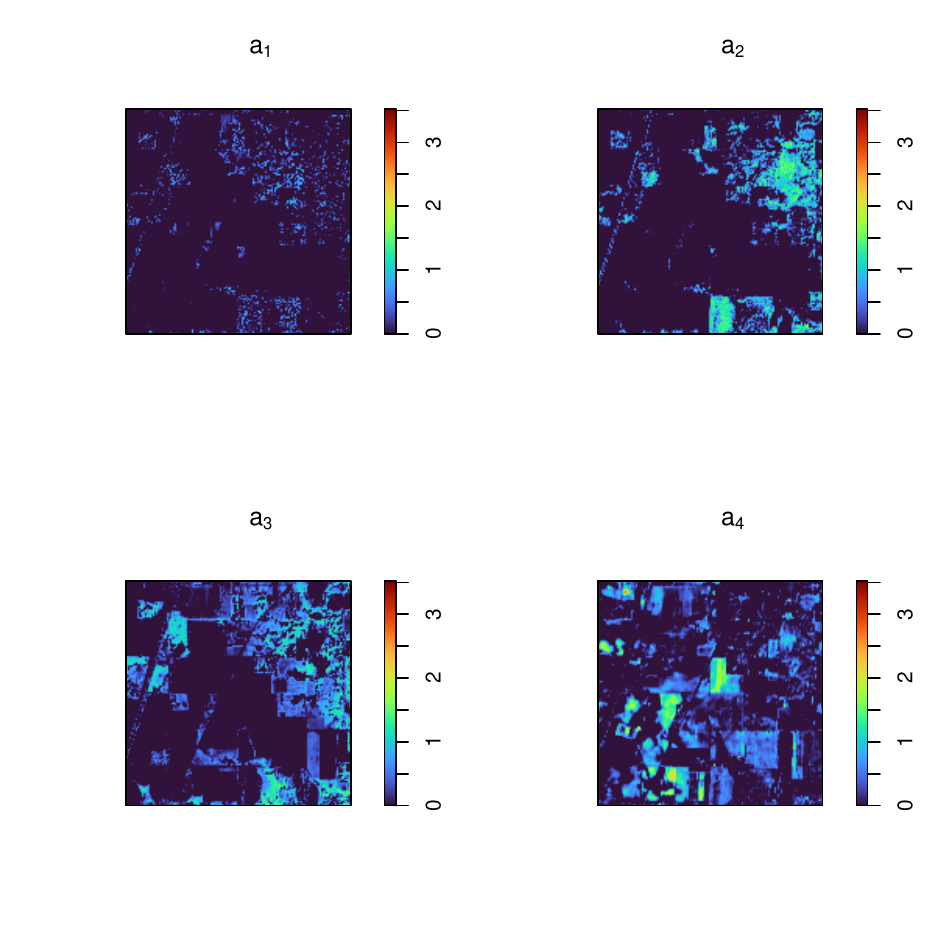}
    \caption{Vertex Component Analysis. In this map, we report the obtained abundance maps, where each abundance $a_j$ is a positive number expressing the linear amplitude of the effect of the endmember $\nu_j$.}
    \label{app:vca_casestudy}
\end{figure}

\begin{figure}[H]
    \centering
    \begin{minipage}{0.48\textwidth}
            \centering
    \includegraphics[width=0.7\linewidth]{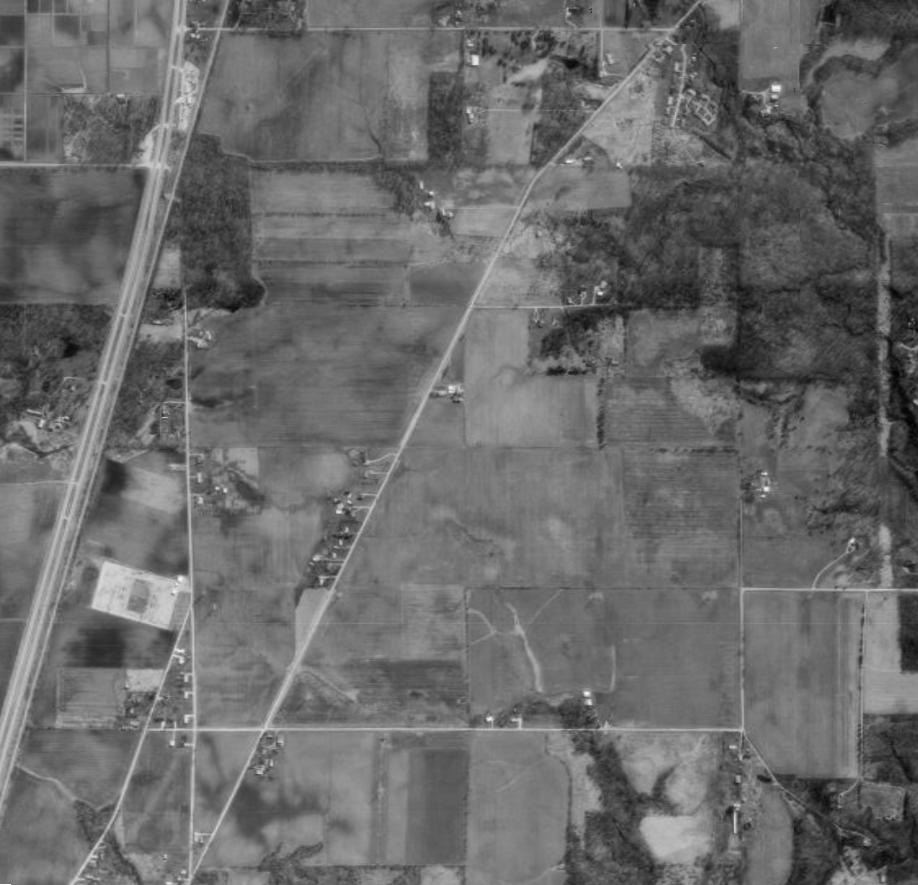}
    \end{minipage}\hfill
    \begin{minipage}{0.48\textwidth}
    \centering
    \includegraphics[width=0.7\linewidth]{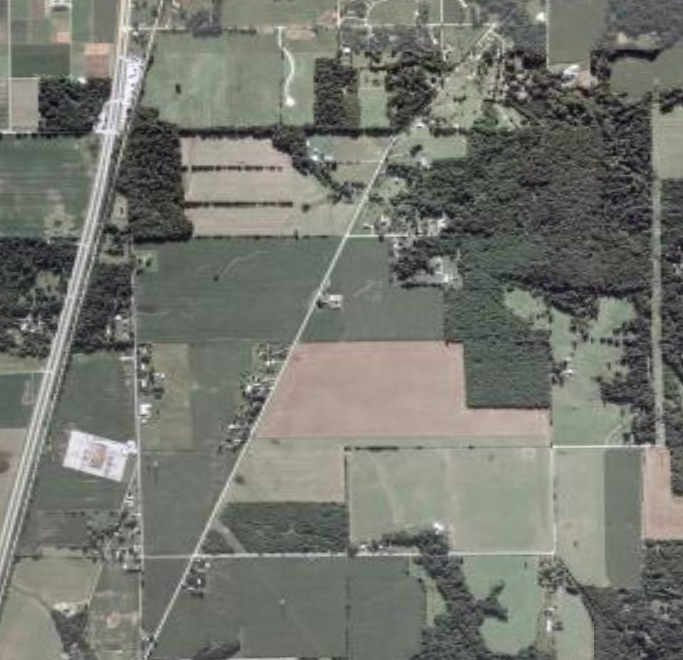}
    \end{minipage}
        \caption{Indian Pines site in 1993 (left), one year after AVIRIS recorded the hyperspectral data. In 1992, we have an almost identical image, with some parts missing. Indian Pines site in 2003 (right) -- the first colored one --, 11 years after AVIRIS recorded the hyperspectral data. This image works as a reference to understand the image in 1993, which is black and white. Source: Google Earth}
        \label{1993_2003}
\end{figure}

\begin{figure}[H]
      \centering    \includegraphics[width=0.6\linewidth]{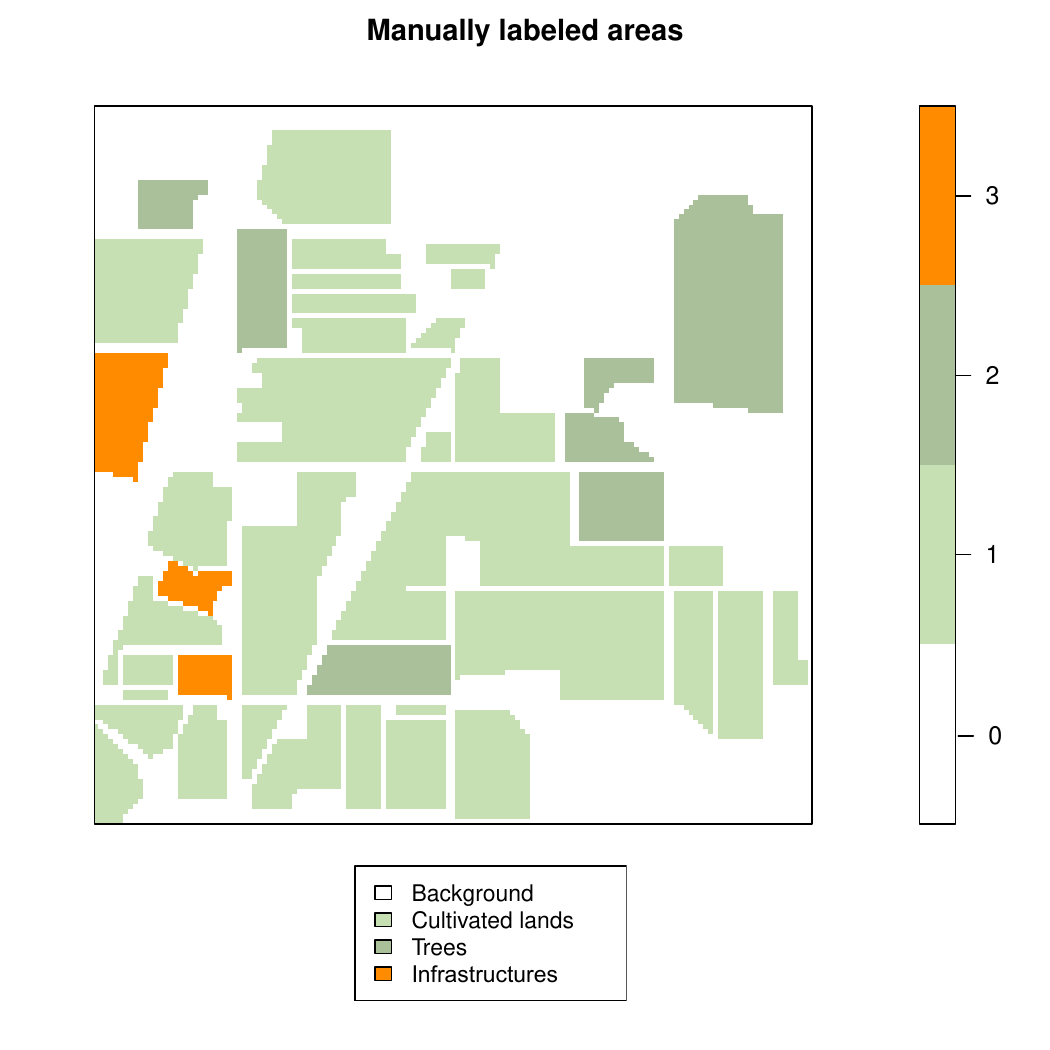}
        \caption{Experts' landcover map with 4 summarising categories. For the detailed subcategories see Figure \ref{manual}.}
      \label{manual_summarised}
\end{figure}

\begin{figure}[H]
\small
        \centering
    \begin{minipage}{0.48\textwidth}
      \centering
    \includegraphics[width=1\linewidth]{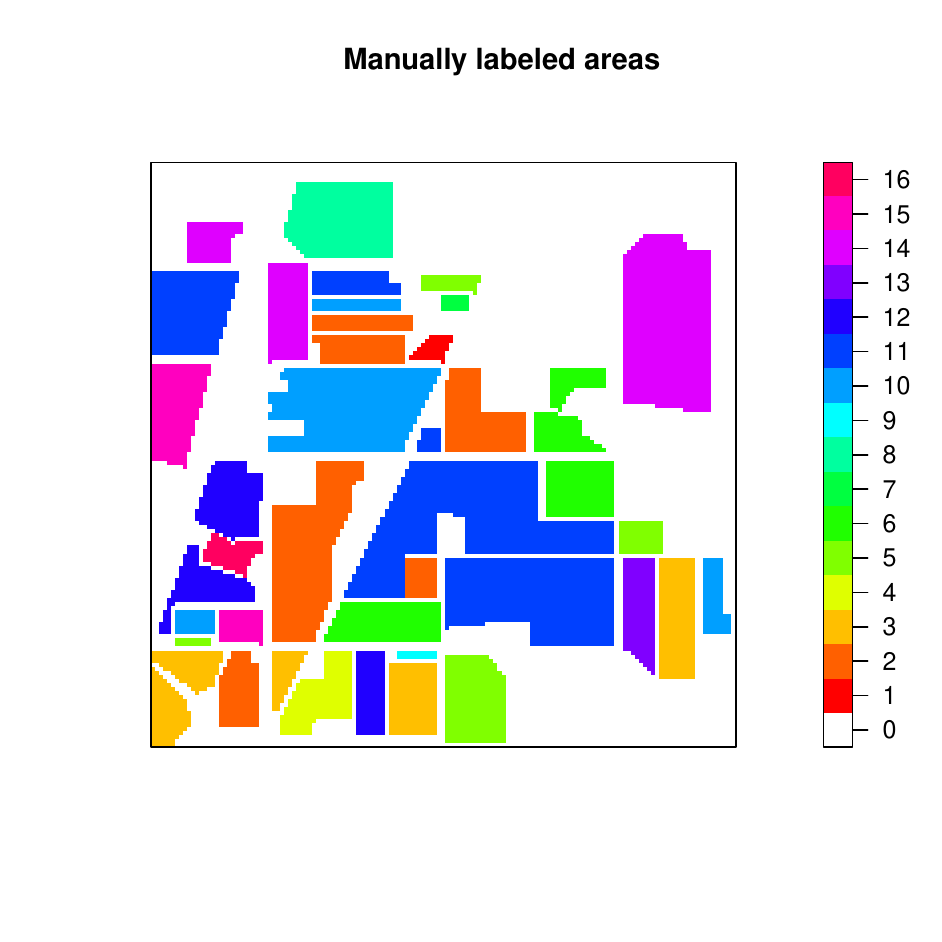}

    \end{minipage}\hfill
\begin{minipage}{0.48\textwidth}
    \centering
    \begin{tabular}{c|c}
        0 & Background \\
        \hline
         1 & Alfalfa \\
        \hline
         2 & Corn-notill \\
        \hline
         3 & Corn-mitill \\
        \hline
         4 & Corn \\
        \hline
         5 & Grass-pasture \\
        \hline
        6 & Grass-trees \\
        \hline
         7 & Grass-pasture-mowed \\
        \hline
         8 & Hay-windrowed \\
        \hline
         9 & Oats \\
        \hline
         10 & Soybean-notill \\
        \hline
         11 & Soybean-mitill \\
        \hline
         12 & Soybean-clean \\
        \hline
         13 & Wheat \\
        \hline
        14 & Woods \\
        \hline
        15 & Buildings-Grass-Trees-Drives \\
        \hline
       16 & Stone-Steel-Towers \\
    \end{tabular}
\end{minipage}
        \caption{Experts' landcover map (left), with legend (right), with the detailed subcategories}
       \label{manual}
\end{figure}
\end{document}